\documentclass[pra,floatfix,footinbib,reprint,superscriptaddress]{revtex4-1}

\usepackage{graphicx}
\usepackage{dcolumn}
\usepackage{times}
\usepackage{bbold}
\usepackage{physics}
\usepackage{bm}
\usepackage{amsthm}
\usepackage{amssymb} 
\usepackage{amsfonts}
\usepackage{tabularx}
\usepackage{subfigure}
\usepackage{tikz}
\usetikzlibrary{matrix}
\usepackage[pdftex, colorlinks, allcolors=blue]{hyperref}

\newtheorem{lemma}{Lemma}
\newtheorem{proposition}{Proposition}

\newtheorem{corollary}{Corollary}

\newcommand{\diagmatdown}
    {\mathbin{\rotatebox[origin=c]{9}{$\diagdown$}}}

\begin{document}

\title{Spectrum and Coherence Properties of the Current-Mirror Qubit}

\author{D.~K.~Weiss}
\email{dkweiss@u.northwestern.edu}
\address{Department of Physics and Astronomy, 
Northwestern University, Evanston, Illinois 60208, USA}
\author{Andy~C.~Y.~Li}
\altaffiliation[Present Address: ]{Fermi National Accelerator Laboratory, Batavia, Illinois 60510, USA}
\address{Department of Physics and Astronomy, Northwestern University, Evanston, Illinois 60208, USA}
\author{D.~G.~Ferguson}
\address{Northrop Grumman Corporation, Linthicum, Maryland 21090, USA}
\author{Jens Koch}
\address{Department of Physics and Astronomy, Northwestern University, Evanston, Illinois 60208, USA}

\date{\today}

\begin{abstract}
The current-mirror circuit [A. Kitaev, arXiv:cond-mat/0609441 (2006)] exhibits a robust ground-state degeneracy and wave functions with disjoint support for appropriate circuit parameters. In this protected regime, Cooper-pair excitons form the relevant low-energy excitations. Based on a full circuit analysis of the current-mirror device, we introduce an effective model that systematically captures the relevant low-energy degrees of freedom, and is amenable to diagonalization using Density Matrix Renormalization Group (DMRG) methods. We find excellent agreement between DMRG and exact diagonalization, and can push DMRG simulations to much larger circuit sizes than feasible for exact diagonalization. We discuss the spectral properties of the current-mirror circuit, and predict coherence times exceeding 1\,ms in parameter regimes believed to be within reach of experiments.
\end{abstract}

\maketitle

\section{Introduction}

Coherence times of superconducting qubits have improved by several orders of magnitude since the seminal experiments with charge qubits at NEC Labs and at Saclay \cite{Nakamura1999,Bouchiat1998}. Nowadays, state-of-the-art transmon and fluxonium qubits have relaxation and dephasing times on the order of 100 $\mu$s \cite{Kjaergaard2019,Devoret2013}. However, many error-correcting codes necessary for a fully fault-tolerant quantum computer, require even longer coherence times \cite{Doucot2012}. A promising avenue for achieving such improvements is implementation of circuits with intrinsic protection from decoherence \cite{Kitaev2006,Brooks2013,Ioffe2002,Gladchenko2009,Bell2014,Groszkowski2018,Smith2019}. One example that has recently gained attention is the $0-\pi$ qubit: it is protected from both pure dephasing and relaxation due to a ground-state degeneracy that is robust with respect to external noisy parameters, and due to qubit states well separated in configuration space \cite{Brooks2013, Dempster2014, Groszkowski2018,DiPaolo2018}. 

Here, we analyze the precursor to the $0-\pi$ qubit, the superconducting current-mirror device, which was introduced by Kitaev \cite{Kitaev2006} as one of the first proposals for an intrinsically protected superconducting qubit. We extend Kitaev's analysis by performing a full circuit analysis of the device, derive an effective model to describe the circuit and discuss in detail the processes that lift the degeneracy. We show that the device possesses enhanced pure dephasing and relaxation times, $T_{\phi}$ and $T_1$, compared to the current state-of-the-art \cite{Kjaergaard2019}. The device, shown in Fig.~\ref{CM}, consists of two linear arrays of Josephson junctions that are capacitively coupled to form a ladder. One end of the ladder is then twisted and connected to the other end, thus producing a M\"obius strip. As prompted by Kitaev, we focus on the regime where the capacitive coupling between the two Josephson junction arrays is large compared to the junction and ground capacitances. In this case (and assuming near-zero offset charges on each node), the low-energy excitations of the circuit consist of Cooper-pair excitons, each formed by a Cooper pair and a Cooper-pair hole across one of the big capacitors. In exciton-hopping among adjacent rungs of the ladder, these two charges move together and generate counter-propagating currents -- hence the name ``current mirror" \cite{Kitaev2006}. Charge excitations other than excitons occur at much higher energies, and will be referred to as agitons. For large circuit size, i.e., in the limit of a large number of big capacitors $N\gg1$, this device exhibits a (near) ground-state degeneracy that is robust with respect to local perturbations. We derive an effective model capturing the low-energy behavior of the circuit, allowing for a quantitative analysis of the spectrum and coherence times of the current mirror.

Our paper is structured as follows. In Sec.~\ref{fullcirc}, we perform a full circuit analysis of the current-mirror device and introduce the exciton/agiton coordinates we use for the remainder of the paper. In Sec.~\ref{effmodelsec} we introduce our effective model, in which we eliminate the high energy degrees of freedom of the circuit. We successfully simulate the resulting Hamiltonian using DMRG techniques in Sec.~\ref{sec:spectrum}, and compare those results to exact diagonalization when applicable. In Sec.~\ref{sec:coherence} we discuss our results for the coherence times of our circuit, and in Sec.~\ref{concsec} we conclude our paper.

\section{Full Circuit Analysis}
\label{fullcirc}

The current-mirror device, shown in Fig.~\ref{CM}, consists of a linear array of $2N$ identical Josephson junctions, wrapped to form the edge of a M\"obius strip. The upper and lower local edges of the strip are connected by capacitors (capacitance $C_B$), forming a series of rungs. In a device with $2N$ junctions, there are $N$ such rungs. Each node in the circuit additionally has a small capacitance to ground (capacitance $C_g$, not shown in Fig.~\ref{CM}).  The Josephson junctions are characterized by their junction capacitance $C_J$ and Josephson energy $E_J=\hbar I_C/(2e)$, where $I_C$ is the junction critical current. The scope of this paper is concerned with the ideal current-mirror circuit in which all junctions, all capacitor rungs, and all ground capacitances are assumed identical, leaving the effects of disorder in circuit elements for future studies.
\begin{figure}
\centering
\includegraphics[width=0.85\columnwidth,angle=0]{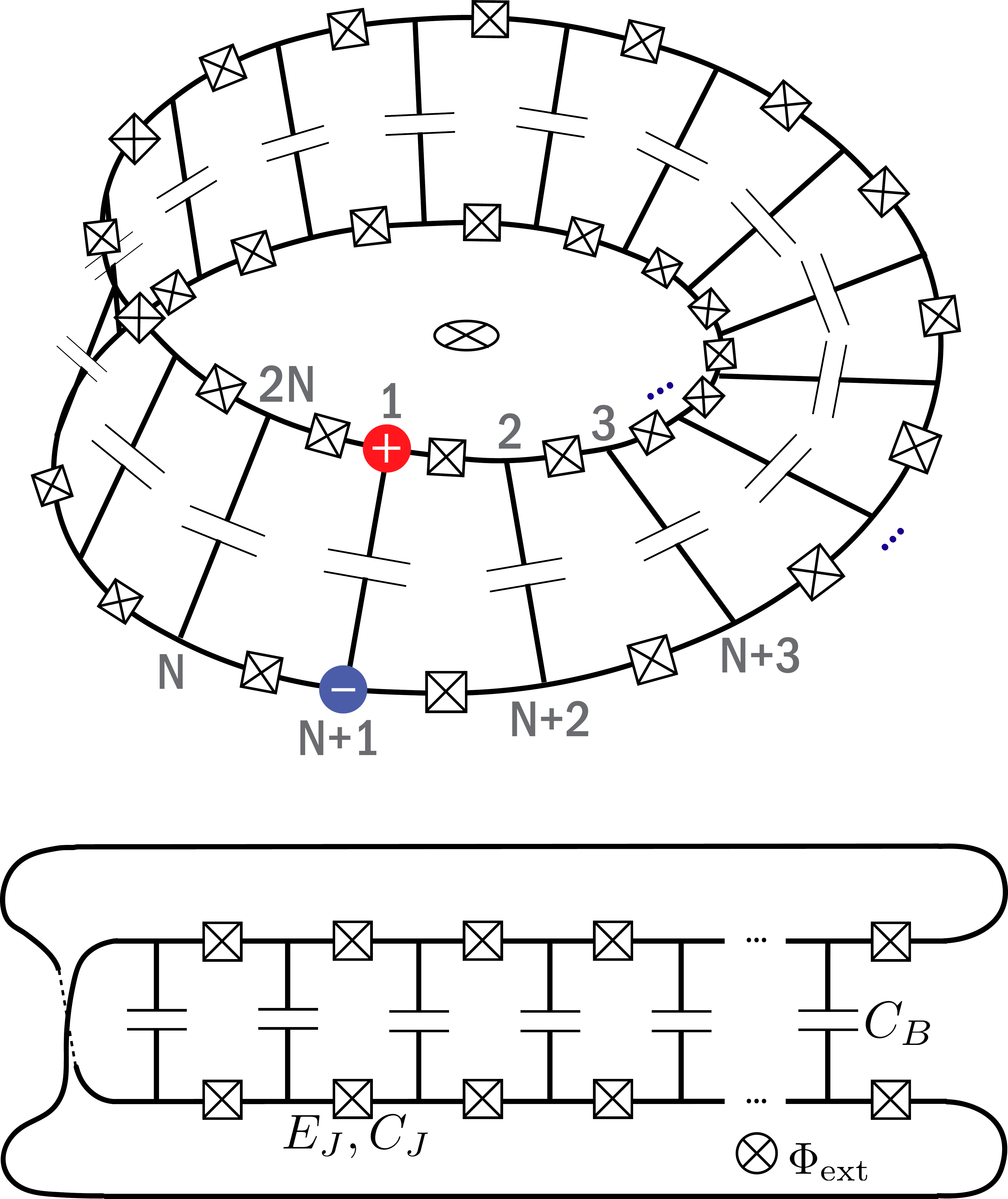}
\caption{\label{CM} The current-mirror circuit, consisting of an array of Josephson junctions ($E_J,\,C_J$) capacitively coupled ($C_B$) to form a M\"obius strip. Low-energy excitations in the intended parameter regime are Cooper-pair excitons, as the one shown on rung 1. An external flux $\Phi_\text{ext}$ penetrating the interior of the M\"obius strip may be used to tune the spectrum of the circuit.}
\end{figure}
The circuit may be described in terms of generalized flux variables $\Phi_{i}$ for each node $1\le i \le 2N$. To simplify notation, we employ reduced flux variables $\phi_i=2\pi\Phi_i/\Phi_0$ where $\Phi_0=h/(2e)$ is the superconducting flux quantum. The circuit Lagrangian for the current mirror is obtained in the standard way \cite{Devoret2004}, and reads
\begin{align}
\label{fullmodelL}
\mathcal{L}&=\frac{1}{2} \left(\frac{\Phi_{0}}{2\pi}\right)^2
\sum_{i,j=1}^{2N}\dot{\phi_i}\mathsf{C}_{ij}\dot{\phi_j}
+\hbar\sum_{j=1}^{2N}\dot{\phi_j}n_{gj} \\ \nonumber
&\quad+ \sideset{}{'}\sum_{j=1}^{2N}E_J\cos(\phi_{j+1}-\phi_{j}-\phi_{\text{ext}}/2N),
\end{align}
where $n_{gj}$ denotes the offset charge associated with node $j$. An external magnetic flux $\Phi_{\text {ext}} = \phi_\text{ext} \Phi_0/2\pi$ may be applied to the interior of the M\"obius strip, as shown in Fig.~\ref{CM}. The primed sum in Eq.\ \eqref{fullmodelL} is understood modulo $2N$, i.e., $\phi_{2N+1}$ is identified with $\phi_1$. Finally, the capacitance matrix $\mathsf{C}$ is given by \footnote{Indices $i,\,j$ are interpreted modulo $2N$.}
\begin{align}
\label{CapMatrix}
\mathsf{C}_{ij}=\begin{cases}
		C_{g}+2C_{J}+C_{B}, &i=j \\
        -C_{J},             &i=j\pm 1 \\
        -C_{B},             &i=j\pm N \\
        0,                  &\textnormal{otherwise}.
\end{cases}
\end{align}

We obtain the Hamiltonian via Legendre transform and quantize it in the usual way by promoting coordinates and conjugate momenta $\phi_{j},n_{j}$ to operators that satisfy the commutation relations $[\phi_{j},n_{k}]=i\delta_{j,k}$ \footnote{The usual caveat applies: due to periodic boundary conditions, $\phi_j$ is strictly speaking ill-defined. A more rigorous, but essentially equivalent way to proceed is to base quantization on the operator $e^{i\phi_j}$.}. This results in the circuit Hamiltonian
\begin{align}
\label{Ham}
H&= \sum_{i,j=1}^{2N}4(\mathsf{E_{C}})_{ij}
    (n_{i}-n_{gi})(n_{j}-n_{gj}) \\ \nonumber
    &\quad -\sideset{}{'}\sum_{j=1}^{2N}E_{J}
    \cos(\phi_{j+1}-\phi_{j}-\phi_{\text{ext}}/2N),
\end{align}
where we have introduced the charging-energy matrix $(\mathsf{E_{C}})_{ij}=e^2\mathsf{C}^{-1}_{ij}/2$. 

Obtaining the spectrum of Eq.~\eqref{Ham} is challenging due to the large number of circuit degrees of freedom. Exact diagonalization in the charge basis is feasible for $N\le 3$ in the relevant parameter regime. For circuit sizes $N>3$, memory requirements for storing the Hamiltonian in the charge basis exceed 1 terabyte when using a charge cutoff of $n_{c}=10$ for each node. The predicted intrinsic protection from relaxation and dephasing, however, specifically requires circuits of large size $N\gg1$ \cite{Kitaev2006}. Therefore, a reduced effective model is needed for a numerical analysis of the circuit's spectrum and coherence properties. Such an effective model can be obtained when focusing on the concrete parameter regime affording decoherence protection, and is further motivated by Ref.~\onlinecite{Lee2003}. In this regime, the circuit is predicted to develop a robust ground-state degeneracy, rendering the circuit insensitive to dephasing channels. Further, the two lowest eigenstates $\ket{0}$ and $\ket{\pi}$ are found to have nearly disjoint support in the multi-dimensional configuration space defined by $\phi_1,\ldots,\phi_{2N}$. As a consequence, all transition matrix elements $\mel{0}{M}{\pi}$ of local operators $M$ are exponentially suppressed, and the circuit is protected from transitions among the computational basis states. 

The protected parameter regime is established by a hierarchical ordering of energy scales. First, the Josephson energy is required to be smaller than the junction charging energy, such that Cooper-pair tunneling can be treated perturbatively. Second, the capacitances $C_B$ are expected to be so large that the associated charging energy forms the smallest charging energy in the hierarchy. The protected parameter regime is thus summarized by the conditions
\begin{equation}
\label{protected}
N\gg1,\quad  E_{J}<E_{C_J},\quad  E_{C_B} < E_{C_J} < E_{C_{g}},
\end{equation}
where $E_{C_a}=e^2/2C_a$. If offset charges vanish and Josephson tunneling is neglected, the lowest-energy excitations are Cooper-pair excitons consisting of a Cooper pair and a Cooper-pair hole \cite{Kitaev2006,Lee2003}, positioned across a big-capacitor rung, see Fig.~\ref{CM} for an example. Exciton charging energies are of the order of $E_{C_B}$. We call non-exciton charge excitations ``agitons". These incur significantly higher charging energies proportional to $E_{C_{J}}$ or $E_{C_{g}}$. The separation into a low-energy exciton subspace and a high-energy agiton subspace is the key ingredient for the development of the effective model to be described next.

\section{Effective Model}
\label{effmodelsec}

The essential idea behind the effective model is to integrate out high-energy agiton excitations in a perturbative treatment of the Cooper-pair tunneling. The tunneling of a single Cooper pair converts an exciton into an agiton. Since the tunneling amplitude $\sim E_J$ is small compared to the agiton charging energy $\sim E_{C_J},\, E_{C_g}$,  such an agiton state will take on the role of a virtual intermediate state before the exciton is restored via a second Cooper-pair tunneling step, see Fig.\ \ref{tunneling_fig}(a).
\begin{figure}
\centering
\includegraphics[width=\columnwidth]{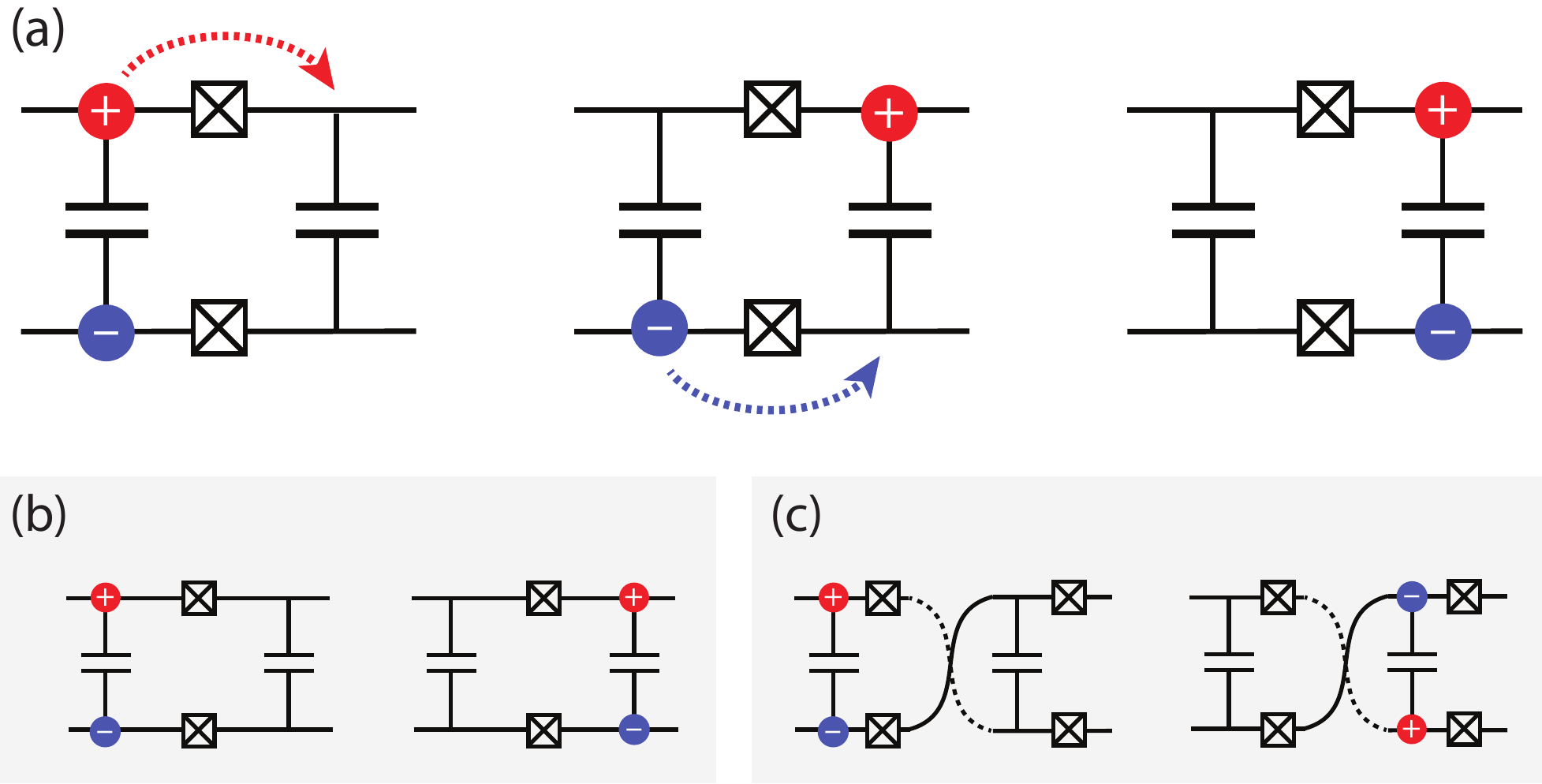}
\caption{\label{tunneling_fig} (a) The tunneling of an exciton from one rung to the next: the Cooper-pair hole tunnels to the right, followed by its partner Cooper pair (or vice versa). (b) Effect of exciton tunneling among regular sites. (c) Exciton tunneling across the stitching point reverses the exciton sign, and corresponds to exciton annihilation on both adjacent rungs. 
}
\end{figure}

\subsection{Exciton and Agiton Variables}
We formalize this idea by introducing appropriate exciton (-) and agiton (+) variables 
\begin{equation}
\label{phipm}
\phi_{j}^{\pm}=\phi_{j} \pm \phi_{N+j},
\end{equation}
and corresponding exciton and agiton charge-number operators
\begin{equation}
\label{npm}
n_{j}^{\pm}=\frac{n_{j}\pm n_{N+j}}{2}.
\end{equation}
Here, the variable index now ranges from $j=1,\ldots, N$. These definitions render commutators among the new operators canonical, i.e.,
\begin{equation}
    [\phi_j^\sigma, n_k^\tau] = i \delta_{jk}\delta_{\sigma\tau},\qquad \sigma,\tau=\pm.
\end{equation}
We note that quantum numbers of the exciton charge and agiton charge operators obey a simple constraint: for each rung $j$ the two charge quantum numbers must be either both integer, or both half-integer. For instance, a single exciton on rung $k$ corresponds to $n_{k}^{-}=1, n_{k}^{+}=0$. If, on the other hand, rung $k$ hosts a single Cooper pair on site $k$ (no charge on site $k+N$), then we have $n_{k}^{-}=-1/2$ and $n_{k}^{+}=-1/2$. The low-energy subspace spanned exclusively by exciton states has quantum numbers $\{n_{j}^{+}=0\}$ and integer-valued $n_{j}^{-}$. 

In terms of exciton and agiton variables, the full circuit Hamiltonian is
\begin{align}
\label{pmHam}
&H = \sum_{\sigma=\pm}\sum_{i,j=1}^{N}4\left(\mathsf{E_{C}^{\sigma}}\right)_{i,j} \left(n_{i}^{\sigma}-n_{gi}^{\sigma}\right) \left(n_{j}^{\sigma}-n_{gj}^{\sigma}\right) \\ \nonumber
&-\sum_{j=1}^{N-1}2E_{J}
    \cos(\tfrac{1}{2} [\phi_{j+1}^{+}-\phi_{j}^{+}-
    \tfrac{\phi_{\text{ext}}}{N}])
    \cos(\tfrac{1}{2}[\phi_{j+1}^{-}-\phi_{j}^{-}])\\\nonumber
&-2E_{J}\cos(\tfrac{1}{2}[\phi_{1}^{+}-\phi_{N}^{+}-
    \tfrac{\phi_{\text{ext}}}{N}])
		\cos(\tfrac{1}{2}[\phi_{1}^{-}+\phi_{N}^{-}]),
\end{align}
where offset charges $n_{gi}$ have been transformed in a way analogous to Eq.\ \eqref{npm}. Note that the potential energy, comprised of the terms in lines 2 and 3, is $4\pi$-periodic in the   $\phi_{j}^{\pm}$ variables. This is a direct result of $n_{j}^{\pm}$ taking on half-integer values. Following this coordinate transformation, the charging energy matrix $\mathsf{E_{C}}$ is brought into block-diagonal form by ordering variables according to $\phi_1^-,\ldots,\phi_N^-,\,\phi_1^+,\ldots,\phi_N^+$:
\begin{equation}
\widetilde{\mathsf{E_{C}}}=\left(\begin{matrix}
                \mathsf{E_{C}^{-}} & 0 \\
                0 & \mathsf{E_{C}^{+}}
\end{matrix}\right).
\end{equation}
Analytical expressions for these exciton and agiton charge matrices $\mathsf{E_{C}^{-}}$ and $\mathsf{E_{C}^{+}}$ are obtained in Appendix~\ref{appendix:ChargeEn}.

To eliminate the high-energy subspace composed of states including agiton excitations, $n_{j}^{+}\not=0$, we employ a Schrieffer-Wolff transformation treating the Josephson part of the Hamiltonian perturbatively. This yields an effective Hamiltonian describing the low-energy subspace involving only the $n_{j}^{-}$ exciton variables.

\subsection{Exciton Tunneling}
The leading-order effect of Josephson tunneling within the exciton subspace is the tunneling of excitons between neighboring rungs, see Fig.\ \ref{tunneling_fig}(a). As a specific example, consider an exciton localized on rung $j$, which is an eigenstate of the unperturbed Hamiltonian ($E_J=0$). Turning on a weak $E_{J}$ allows for Cooper pairs to tunnel. The tunneling of a single Cooper pair or Cooper-pair hole to rung $j+1$ incurs a high charging-energy cost due to biasing of the junction capacitor. We can compute the relevant charging energies from Eq.\ \eqref{pmHam}. The initial single-exciton state and intermediate agiton state are $\ket{\text{e}} = |n_j^-=1\rangle$, and
\begin{equation}
    \ket{\text{a}\pm} = |n_j^-=n_{j+1}^-=\tfrac{1}{2},\, n_j^+=\mp\tfrac{1}{2},\, n_{j+1}^+=\pm\tfrac{1}{2}\rangle,
\end{equation}
where $\pm$ corresponds to tunneling of a Cooper-pair hole or a Cooper pair, respectively, and we are only showing the non-zero charge quantum numbers. The corresponding charging energies are
$E_\text{e}=\mathcal{O}(E_{C_B})$ and 
\begin{align}
\label{Edenom}
&E_\text{a}^\pm = 
    2(E_{{\text C}0}^{+}-E_{{\text C}1}^{+})\\ \nonumber
    &\pm \sideset{}{'}\sum_{m=1}^{N}4\left[(\mathsf{E_{C}^{+}})_{m,j}-
        	(\mathsf{E_{C}^{+}})_{m,j+1}\right]n_{gm}^{+} + \mathcal{O}(E_{C_B}),
\end{align}
Here, we have defined $E_{{\text C}0}^{+}=(\mathsf{E_{C}^{+}})_{j,j}$ and $E_{{\text C}1}^{+}=(\mathsf{E_{C}^{+}})_{j,j\pm 1}$; see Appendix~\ref{appendix:ChargeEn} for explicit expressions of the charging-energy matrix elements $(\mathsf{E_{C}^{+}})_{j,k}$. Since $E_\text{a}^\pm=\mathcal{O}(E_{C_J},E_{C_g})$, we can approximate the energy cost by the agiton contribution alone, $\Delta E_j^\pm=|E_\text{e}-E_\text{a}^\pm|\approx E_\text{a}^\pm$.

The effective small parameters governing the perturbation theory are thus $E_{J}/\Delta E_{j}^{\pm}$. Through the second-order process depicted in Fig.~\ref{tunneling_fig}(a), an exciton can tunnel from one rung to the neighboring rungs. In such a process, the high-energy agiton subspace ($n_{j}^{+}\neq0$) is only accessed via intermediate virtual states. The resulting second-order effective Hamiltonian in the exciton subspace  is
\begin{align}
\label{effHam}
H^{(2)}=&\sum_{i,j=1}^{N}4
    \left(\mathsf{E_{C}^{-}}\right)_{i,j}
    \left(n_{i}^{-}-n_{gi}^{-}\right)
    \left(n_{j}^{-}-n_{gj}^{-}\right) \\ \nonumber
-&\sum_{j=1}^{N-1}J_{j}\cos(\phi_{j+1}^{-}-\phi_{j}^{-})-
J_{N}\cos(\phi_{1}^{-}+\phi_{N}^{-}),
\end{align}
where the coefficients $J_{j}$ are the exciton tunneling rates, see Appendix \ref{appendix:Nthorder} for details. In the protected regime and for vanishing offset charges, the exciton tunneling rates have the uniform expression 
\begin{equation}
\label{simpleJ}
J=\frac{E_{J}^2}{2E_{C_{J}}}
    +\textstyle\mathcal{O}\left(\frac{C_{g}}{C_{J}},\frac{1}{N}\right),
\end{equation}
which is intuitive given the two-step hopping process ($\sim E_{J}^2$) with intermediate state energy penalty ($\sim 2E_{C_{J}}$).

The sign deviation in the final exciton hopping term of Eq.\ \eqref{effHam} occurs as a direct consequence of the M\"obius topology. Its origin can be understood  with the help of Fig.~\ref{tunneling_fig}, depicting the exciton tunneling process across the twist point. Generally, exciton tunneling is understood as exciton annihilation on one rung and creation on a neighboring rung, see Fig.~\ref{tunneling_fig}(b). However, tunneling across the twist point results in either exciton annihilation or exciton creation on both rungs, see Fig.~\ref{tunneling_fig}(c). We note that the specific location of the twist point is irrelevant to the physics, as variables may be cyclically permuted. 

Inspection of the charging energies for excitons, $\mathsf{E_{C}^{-}}$, shows that the charge-charge interaction is relatively short-ranged for excitons in the protected regime. Using analytical results from Appendix \ref{appendix:ChargeEn}, we find the asymptotic expressions
\begin{align}\nonumber
&\left(\mathsf{E_{C}^{-}}\right)_{j,j}= \frac{e^2}{2C_{B}}+\textstyle\mathcal{O}\left(\frac{C_{J}}{C_{B}},\frac{C_{g}}{C_{B}},\frac{1}{N}\right), \\ \label{ECminus}
&\left(\mathsf{E_{C}^{-}}\right)_{j,j+1}= \frac{e^2C_{J}}{4C_{B}^2} +\textstyle\mathcal{O}\left(\frac{C_{J}}{C_{B}},\frac{C_{g}}{C_{B}},\frac{1}{N}\right), \\\nonumber
&\left(\mathsf{E_{C}^{-}}\right)_{1,N}= -\frac{e^2C_{J}}{4C_{B}^2} + \textstyle\mathcal{O}\left(\frac{C_{J}}{C_{B}},\frac{C_{g}}{C_{B}},\frac{1}{N}\right).
\end{align}
All other off-diagonal elements are strongly suppressed in higher powers of $C_{J}/C_{B}$.

We emphasize that the absence of the external flux $\phi_{\text{ext}}$ from the effective Hamiltonian \eqref{effHam} is not an error of omission. Remarkably, within the second-order computation the flux drops out exactly, see Appendix~\ref{appendix:Nthorder} for details. This indicates weak sensitivity of the circuit to flux, which is only re-established by higher-order terms, as we will demonstrate below. 

Integrating out the agiton degrees of freedom restores $2\pi$-periodicity of the potential energy in Eq.\ \eqref{effHam}, a direct consequence of the elimination of half-integer eigenvalues of $n_j^-$ in the exciton subspace. Inspection reveals that the potential energy $V(\phi_1^-,\ldots\phi_N^-)$ has the form of an $N$-dimensional double-well potential, with minima located at $\{\phi_{j}^{-}=0\}$ and $\{\phi_{j}^{-}=\pi\}$. The values of the two potential minima are identical, thus providing the basis for the (near-)degeneracy of ground and first excited eigenstates of the current-mirror circuit, as envisioned by Kitaev \cite{Kitaev2006}. Note that exact degeneracy of the minima leads to ground- and first-excited-state wave functions that lack disjoint support. Higher-order processes that lead to degeneracy breaking are therefore crucial for the current-mirror qubit to operate as intended. 

An important point to note in our application of perturbation theory is the assumption that charge frustration can be neglected. This can be seen explicitly in Eq.~\eqref{Edenom}, where energy denominators approach zero for $n_{gj}^{+}=1/2$. However, as noted in Refs.~\onlinecite{Lee2003,Kitaev2006}, while the nature of the low-energy excitations changes from Cooper-pair excitons to Cooper-pair and void excitons, the overall behavior of the circuit is not expected to change. 

\subsection{Degeneracy-Breaking Terms}

The leading-order processes that lift the degeneracy among the two potential minima correspond to the annihilation or creation of an odd number $m$ of excitons ($1\leq m \leq N$), see Fig.\ \ref{excitoncreation}. The circular circuit representation employed in that figure is topologically equivalent to the original M\"obius circuit, where it is understood that capacitive connections do not ``touch" each other at the center. The creation or annihilation of such excitons requires $N$-th order perturbative processes, as we show via a map to the so-called assignment problem on the circle \cite{Werman1986,Karp1975}, see Appendix ~\ref{appendix:KNsteps}. There, we further prove that leading-order exciton creation and annihilation leads to sign alternation, i.e., neighboring Cooper-pair charges alternate signs when moving around the circle, see Fig.\ \ref{excitoncreation}. Any process leading to the creation of an odd number of excitons that does not obey sign alternation is of higher order, and is therefore subdominant.

The derivation of the associated effective-Hamiltonian terms is discussed in Appendix~\ref{appendix:Nthorder}, and proceeds by taking the Schrieffer-Wolff transformation to $N$-th order.
\begin{figure}
    \centering
    \includegraphics[width=1.0\columnwidth]{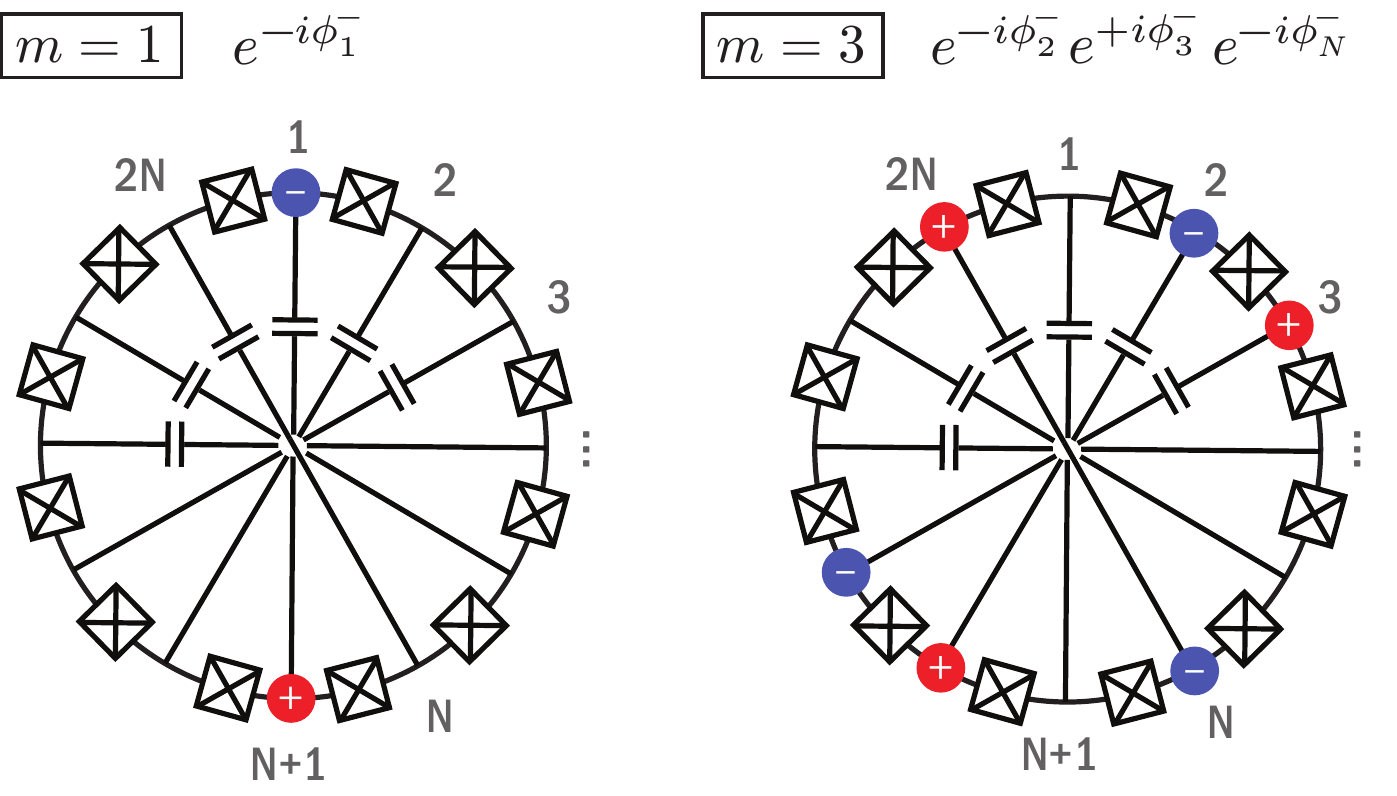}
    \caption{Creation of an odd number of excitons in the current-mirror circuit, here in its equivalent circular representation. Generation of a single exciton ($m=1$) on rung $j=1$ is achieved by the operator $e^{-i\phi_{1}^{-}}$. Likewise, $m=3$ excitons with alternating signs are generated on rungs $j=2,3,N$ by the operator $e^{-i\phi_{2}^{-}}e^{i\phi_{3}^{-}}e^{-i\phi_{N}^{-}}$. The figure shows the configurations obtained when applying these operators to the charge-neutral circuit. }
    \label{excitoncreation}
\end{figure}
This way, we obtain the following expression for the degeneracy-breaking terms:
\begin{align}
\label{KNcirc}
H_{\text{K}} =
-K\cos({\textstyle\frac{\phi_{\text{ext}}}{2}})\sum_{m\leq N}^{\text{odd}}\,\sum_{i_{1}<\cdots<i_{m}}
\cos\bigg[\sum_{j=1}^{m}(-1)^{j}\phi_{i_{j}}^{-}\bigg],
\end{align}
Here, $K$ is the rate at which excitons of the described kind are generated and annihilated. The summation index $m$ of the outer sum counts the number of created/annihilated excitons, and runs over all odd integers in the range of $1$ through $N$. The inner sum accounts for all possible positions $i_{j}$ of the $m$ excitons. Sign alternation of charges is reflected by the corresponding prefactor in the cosine argument. We observe that the exciton-generating terms also re-establish dependence of the spectrum on the external flux. This dependence remains strongly suppressed, however, due to the smallness of the rate $K$.

Setting a rigorous bound on $K$ itself is generally nontrivial for the combination of two reasons. First, the number of contributing perturbative paths is large, namely $\sim N!$; second, the ubiquitous energy denominators intricately depend on the energies of the various high-energy virtual states accessed in the course of the process.
One can obtain the upper bound
\begin{equation}
2^{N-1}K< E_{J}\left(\frac{2E_{J}}{\Delta E}\right)^{N-1},
\end{equation}
where $\Delta E=E_{{\text C}0}^{+}-E_{{\text C}1}^{+}$, and $2^{N-1}$ is the number of  degeneracy-breaking terms in Eq.~\eqref{KNcirc}, see Appendix~\ref{appendix:Nthorder} for details. In the parameter regime studied in this paper, see Tab.~\ref{table:parameters}, the numerator $2E_{J}=38$ GHz is smaller than the denominator $\Delta E$ which increases as a function of $N$, with $\Delta E>60$ GHz when $N\geq 4$. Evidently, creation and annihilation of excitons and, hence, the lifting of the potential-minima degeneracy, is exponentially suppressed with growing circuit size $N$.

The resulting full effective Hamiltonian, capturing both exciton tunneling and degeneracy breaking to leading order, is given by
\begin{align}
\label{MobHam}
&H_{\text {eff}}=\sum_{i,j=1}^{N}4
    \left(\mathsf{E_{C}^{-}}\right)_{i,j}
    \left(n_{i}^{-}-n_{gi}\right)
    \left(n_{j}^{-}-n_{gj}^{-}\right) \\ \nonumber
    &\qquad-\sum_{j=1}^{N-1}J\cos(\phi_{j+1}^{-}-\phi_{j}^{-})-
    J\cos(\phi_{1}^{-}+\phi_{N}^{-}) \\ \nonumber
    &\qquad-K\cos(\textstyle\frac{\phi_{\text{ext}}}{2})\sum_{m\leq N}^{\text{odd}}\sum_{i_{1}<\cdots<i_{m}}
    \cos\bigg[\sum_{j=1}^{m}(-1)^{j}\phi_{i_{j}}^{-}\bigg].
\end{align}
The degeneracy of the two potential minima located at $\{\phi_{j}^{-}=0\}$ and $\{\phi_{j}^{-}=\pi\}$ is now weakly broken by $H_{\text{K}}$. In the parameter regime of interest, $H_{\text{K}}$ is exponentially suppressed, and minima remain  near-degenerate. Additionally, the kinetic energy scale $E_{C_{B}}=0.2$ GHz is small compared the barrier height $J\approx 2$ GHz, leading to very little tunneling between the two minima. We thus obtain localized, nearly degenerate states for the ground and first-excited states, corresponding to the qubit manifold.

\subsection{Linearization of the Effective Model}

Low-energy excitations of the current-mirror circuit arise as harmonic excitations within the two potential wells of the effective model. To make this statement quantitative, we extract the normal-mode frequencies associated with the linearized version of the second-order effective Hamiltonian \eqref{effHam}. (Degeneracy-breaking terms proportional to $K$ can be safely neglected in this context.)

The normal-mode analysis is conveniently performed in the Lagrangian picture, where Taylor expansion of the potential energy about the $\{\phi_{j}^{-}=0\}$ and $\{\phi_{j}^{-}=\pi\}$ minima yields
\begin{align}
\mathcal{L}_{\text {eff}}^{0,\pi}&=\frac{1}{2}
\left(\frac{\Phi_0}{2\pi}\right)^2\sum_{i,j}^{N}
			\dot{\phi}_{i}^{-}(\mathsf{C}_{-})_{i,j}
			\dot{\phi}_{j}^{-} \\ \nonumber
            &\quad+\sum_{j=1}^{N-1}\frac{J}{2}
            (\phi_{j+1}^{-}-\phi_{j}^{-})^2
            +\frac{J}{2}(\phi_{1}^{-}+\phi_{N}^{-})^2.
\end{align}
We seek normal-mode solutions of the form $\vec{\phi}^- = \vec{\xi}_k e^{i\omega_k t}$. As usual, plugging this ansatz into the equation of motion yields the generalized eigenvalue problem $\mathsf{M}\vec{\xi}_k = \omega_k^2 \mathsf{C}_- \vec{\xi}_k$ for the normal-mode amplitudes $\vec{\xi}_k$ and associated eigenfrequencies $\omega_k$. Here, $\mathsf{M}$ denotes the coefficient matrix for the potential bilinear form. Inspection shows that $\mathsf{M}$ and the capacitance matrix $\mathsf{C}_-$ are both real, symmetric, and tridiagonal with additional corner elements. For this reason, they possess the same system of orthonormal eigenvectors, $\mathsf{M}\vec{\xi}_k = m_k \vec{\xi}_k$ and $\mathsf{C}_-\vec{\xi}_k = c_k \vec{\xi}_k$, and the eigenfrequencies are obtained from $\omega_k^2=m_k/c_k$ with the result  
\begin{align}\nonumber
\omega_{k}&=\frac{2\pi}{\Phi_0}
    \sqrt{\frac{4J\sin^2\left[(2k-1)\pi/ 2N\right]}
			{C_{B}+C_{g}/2+2C_{J}\sin^2
			\left[(2k-1)\pi/2N\right]}} \\ 
		&= \frac{2\pi}{\Phi_0}\sqrt{\frac{4J}{C_{B}}}
    \left|\sin \frac{(2k-1)\pi}{2N} \right|
    +\mathcal{O}\left(\textstyle \frac{C_{g}}{C_{B}},
    \frac{C_{J}}{C_{B}}\right).
\label{modefreqs}
\end{align}
Normal-mode frequencies are generally four-fold degenerate due to two factors: first, the $0$ and $\pi$ minima contribute equal sets of normal modes; second, one finds $\omega_{k}=\omega_{N-k+1}$ for all $k\le \lfloor N/2 \rfloor $. (In the case of odd $N$, the highest normal-mode frequency $\omega_{(N+1)/2}$ is only two-fold degenerate, while the rest remain four-fold degenerate.)

A key insight from Eq.\ \eqref{modefreqs} is that the eigenfrequencies of the lowest-lying modes scale with $1/N$. Therefore, the circuit size of the current-mirror qubit should not be chosen too large, in order to avoid unwanted thermal population of low-lying excited states. We will argue in Sec.\ \ref{sec:coherence}  that there is indeed a trade-off between depolarization and dephasing times which have opposite behavior as a function of circuit size $N$.
\begin{figure*}
    \centering
    \includegraphics[width=0.78\textwidth]{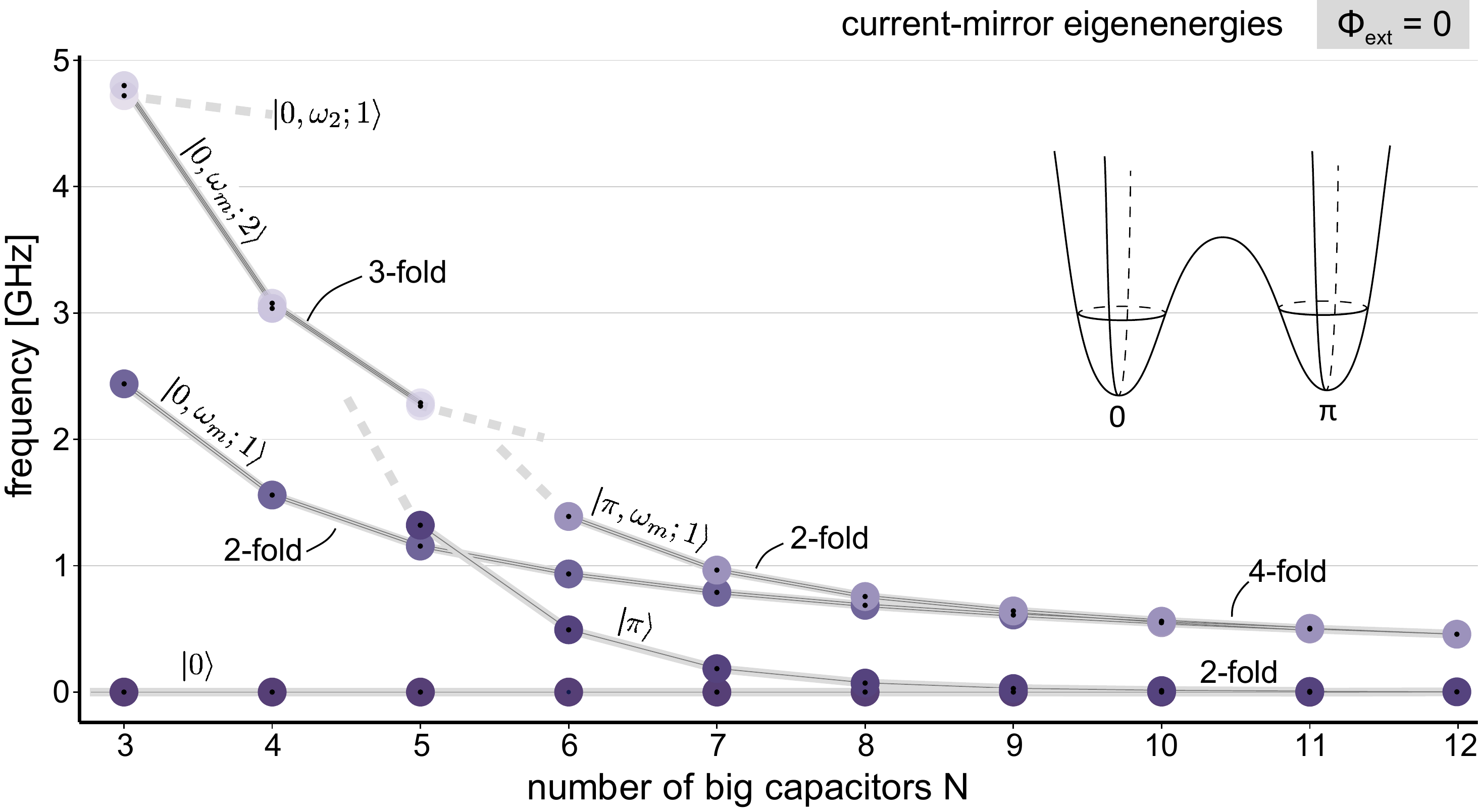}
    \caption{\label{Heffspectra} DMRG results: six lowest eigenenergies of the current-mirror circuit as a function of circuit size $N$ (number of big capacitors), based on the effective Hamiltonian [Eq.\ \eqref{MobHam}]. For $N\ge7$, ground state and first excited state, $|0\rangle$ and $|\pi\rangle$, are the nearly degenerate, lowest eigenstates localized in the $0$ and $\pi$ wells. Other eigenstates correspond to harmonic excitations within the two wells, and are denoted  $\ket{\alpha,\omega_{k};n}$, where $\alpha=0,\pi$ labels the well, $k$ is the mode index [see Eq.~\eqref{modefreqs}], and $n$ the number of harmonic excitations. $m$ is placeholder for the degenerate modes $\omega_0$ and $\omega_1$. The inset shows a schematic 2d projection of the potential energy with minima at $0$ and $\pi$. (Circuit parameters used: see Tab.~\ref{table:parameters}.)
    }
\end{figure*}
\subsection{Nature of Ground-State Degeneracy in the Full Model}
We have successfully confirmed the near-degeneracy of the lowest two eigenstates within the effective model. It is instructive to revisit the origin of this degeneracy in the context of the full circuit model [Eq.~\eqref{Ham}]. We expect low-lying eigenstates to be related to the minima of the potential energy $U=-\sum_{j=1}^{2N}E_{J}\cos(\phi_{j+1}-\phi_{j}-\phi_{\text{ext}}/2N)$. Finding its local minima via $\nabla U=0$ leads to a set of $2N$ modular equations,
\begin{align}\nonumber
\phi_{1}-\phi_{2N}&\equiv\phi_{2}-\phi_{1}, \\ \label{phasediff}
\phi_{j+1}-\phi_{j}&\equiv\phi_j-\phi_{j-1} 
    &2\le j \le 2N-1, \\ \nonumber
\phi_{2N}-\phi_{2N-1}&\equiv\phi_{1}-\phi_{2N},
\end{align}
where all congruences are modulo $2\pi$. 
At the minima locations, phase differences between adjacent nodes are thus identical to some constant value $\Delta\phi$. It is convenient to make use of gauge freedom and set $\phi_{1}=0$, to obtain $\phi_{j}=(j-1)\Delta \phi$, where $1\leq j\leq 2N-1$. Solving for $\phi_{2N}$ yields $\Delta\phi=\pi l/N$ with integer $l$ denoting the number of phase windings or vortices. In summary, the local potential minima are labeled by the integer vortex number $l$, and have coordinates
\begin{equation}
\label{phimin1}
(\vec{\phi}_l)_j =\frac{\pi l}{N}(j-1), \qquad |l| < \frac{N}{2}+\frac{\phi_{\text{ext}}}{2\pi}.
\end{equation}
As opposed to the effective-Hamiltonian potential with only two minima at $0$ and $\pi$, one faces a multitude of local minima in the full-circuit model. There is a simple correspondence between minima in the two models, namely
\begin{align}
    \vec{\phi} &= \vec{\phi}_{l} \quad (l=\substack{\text{even}\\ \text{odd}}) \qquad \leftrightarrow \qquad \{\phi_j^-=\substack{0\\\pi}\},
\end{align}
i.e., minima with even (odd) vortex parity contribute to the effective-model minimum at $0$ ($\pi$). In particular, the ground-state wave function occupies even-$l$ minima, $l=0,\pm 2,\ldots$, while the nearly degenerate first excited state occupies odd-$l$ minima, $l=\pm1,\pm3,\ldots$ 

Whether low-lying wave functions span multiple minima depends on whether tunneling between these minima is significant. We thus examine the eigenvalues and eigenvectors of the effective-mass tensor governing the tunneling dynamics, here given by the capacitance matrix $\mathsf{C}$ [Eq.\ \eqref{CapMatrix}]. 
The eigenvectors of $\mathsf{C}$ are the real and imaginary parts of 
\begin{equation}
(\vec{\zeta}_{k})_j = e^{i\pi j k/N}/\sqrt{2N}
\end{equation} 
with corresponding eigenvalues 
\begin{align}
\gamma_k = C_{g}+4C_{J}\sin^2 \left(\frac{\pi k}{2N}\right)+[1-(-1)^k]C_{B}.
\end{align}
Here, $k$ ranges from $1$ to $2N$, and eigenvalues are generally two-fold degenerate (with the exceptions of non-degenerate $k=N$ and $k=2N$). 

From the eigenvalues we infer that the effective masses for even $k$ involve only the small capacitances $C_{J}$ and $C_{g}$ (light effective mass), while odd-$k$ eigenvalues involve the large capacitance $C_{B}$ (heavy effective mass). At the same time, it is only along the directions of odd-$k$ eigenvectors that the values of $\phi_j^-$ variables change, since $(\vec{\zeta}_{k})_{j}=-(\vec{\zeta}_{k})_{j+N}$ for odd $k$. Hence, tunneling between minima of different vortex parity is strongly suppressed, confirming the picture of low-lying eigenstates having definite vortex parity.

\section{Spectrum of the Current Mirror}
\label{sec:spectrum}

\begin{table}
\centering
\begin{tabularx}{\columnwidth}{*4{>{\centering\arraybackslash}X}}
\hline\hline 
 $E_{C_{B}}$ &$E_{C_{J}}$ &$E_{C_{g}}$ & $E_{J}$ \\
 \hline
    0.2 $h\cdot$GHz  & 100 $h\cdot$GHz  &200 $h\cdot$GHz & 19   $h\cdot$GHz  \\ \hline\hline
\end{tabularx}
\caption{Circuit parameters used in the numerical analysis, consistent with the conditions \eqref{protected} for the protected regime.}
\label{table:parameters}
\end{table}

The simplification achieved with the effective model makes the problem of finding the current-mirror spectrum amenable to numerical diagonalization. The selected current-mirror parameters are given in Tab.~\ref{table:parameters}. This choice firmly places the system in the intended protected regime where the effective model is valid. We note that from the experimental perspective, achieving the small junction capacitance necessary to realize a device with these parameters will likely be the biggest challenge. We compute the eigenspectrum of the effective model as a function of circuit size (number of big capacitors $N$) as well as external flux $\Phi_{\text{ext}}$, using two separate techniques: exact diagonalization and Density Matrix Renormalization Group (DMRG).

We first perform exact diagonalization in the exciton-number basis. We employ the simplest truncation scheme by choosing an appropriate exciton-number cutoff $n_c$ for each rung, thus retaining exciton-basis states $-n_{c}<n_{j}^{-}<n_{c}$. (We find that $n_{c}=10$ is sufficient for convergence in our case.) This way, exact diagonalization is feasible up to circuit size $N=6$, beyond which memory requirements become excessively large \footnote{A truncation scheme using a global exciton-number cutoff will likely succeed in pushing exact diagonalization to slightly larger system sizes, but was not our main interest here.}.

To extend our numerical treatment to current-mirror circuits of sizes $N>6$, we make use of DMRG methods. These methods have been employed very successfully in simulating one-dimensional quantum systems by an efficient Hilbert-space truncation that only retains largest-weight eigenvectors of the system density matrix at each step in the algorithm \cite{Schollwock2011,White1992,*White1993}. A large class of many-body and spin systems intractable with exact diagonalization can be handled with DMRG. In the context of superconducting circuits, DMRG has previously facilitated the study of capacitively-coupled Josephson junction necklaces  by Lee et al.\ \cite{Lee2003}.
Since memory efficiency and fast convergence of DMRG algorithms generally rely upon the short-ranged nature of interactions \cite{Schollwock2011}, Lee and coworkers eliminated  long-ranged capacitive interactions by neglecting junction capacitances. They further applied open boundary conditions  which are known to speed up convergence  \cite{Lee2003, Schollwock2011, White1992, *White1993}.

We proceed in a similar way for the current-mirror circuit, noting that capacitive interactions between rung degrees of freedom are relatively short-ranged in the effective model, see Eq.\ \eqref{ECminus}. The long-range interactions produced by the degeneracy-breaking term $H_{K}$ are known to be weak and do not impede the treatment. The effective Hamiltonian $H_\text{eff}$ \eqref{effHam} thus essentially describes a one-dimensional exciton model with nearest-neighbor hopping suitable for the DMRG algorithm. The use of open boundary conditions, however, is not appropriate in our case, as the M\"obius topology of the current-mirror circuit is crucial for the ground-state degeneracy. We have implemented DMRG for the effective-model Hamiltonian  using the ITensor package developed by Stoudenmire and White \cite{Stoudenmire}.  In the following, we present DMRG results for circuit sizes up to $N=12$.  To assess the accuracy of DMRG spectra, we have compared against effective-model spectra obtained with exact diagonalization up to 6 big capacitors. We find excellent agreement, with relative deviations less than $2\times10^{-6}$. (We have confirmed that even larger circuit sizes can be tackled with DMRG, and the computational bottleneck is no longer memory but runtime.)

The current mirror has been predicted to exhibit ground-state degeneracy in the limit of large $N$ \cite{Kitaev2006}. For experimental realizations, it is a pertinent question what concrete circuit size is required to enter the protected regime. Results from our DMRG calculations shed new light on this issue.  Figure \ref{Heffspectra} shows the lowest six effective-model eigenenergies versus circuit size $N$, for vanishing magnetic flux. We observe that ground and first excited states rapidly approach each other above $N=6$, consistent with the exponential suppression of the degeneracy-breaking terms. Labeling of qubit states in Fig.\ \ref{Heffspectra} is guided by the linearized effective model: the two qubit states eventually becoming degenerate are denoted by $\ket{0}$ and $\ket{\pi}$ in reference to the corresponding potential minima. Other low-lying eigenstates can be identified as harmonic excitations in the two wells with excitation energies approximated by multiples of the mode frequencies from Eq.~\eqref{modefreqs}. A state in the $\alpha$ well ($0$ or $\pi$) with $n$ excitations in mode $k$ is written as $|\alpha,\omega_k;n\rangle$.

For $N=3$, the $\pi$ well minimum is still significantly above the $0$ well minimum, and the latter hosts three eigenmodes with frequencies $\omega_{0}=\omega_{1}<\omega_{2}$. Coincidentally, $\omega_{2}/\omega_{0}\approx2$, leading to the apparent near-degeneracy for the highest energy eigenstates shown at $N=3$. Only the lowest mode frequencies scale with $1/N$, leading to the disappearance of states involving $\omega_{2}$ from the low-energy spectrum above $N=3$. At large $N$, the degeneracy-breaking terms in Eq.~\eqref{MobHam} become exponentially small and the low-energy eigenstates are well described as harmonic excitations in one of the two wells, where each well is itself a 2-dimensional harmonic oscillator with two-fold degenerate mode frequencies.

We have further investigated the case of half-integer flux, $\phi_\text{ext}=\pi$, where degeneracy-breaking terms in Eq.~\eqref{effHam} vanish. Consistent with that, the ground state is nearly degenerate already at $N=3$. However, operation at this point produces a completely symmetric double-well potential, thus leading to eigenstates that are symmetric and anti-symmetric superpositions of the localized wave functions in each well. These qubit states are not protected from relaxation by disjoint-support arguments, defeating one of the original purposes for considering the superconducting current mirror as a qubit.

Our numerical results have quantified the circuit size required in order to enter the protected parameter regime. For the parameters we have chosen, the qubit states become the lowest-energy eigenstates of the circuit only for $N\geq6$. Below this, the degeneracy-breaking terms are large enough to push the qubit state in the $\pi$ well above the low-energy excitations in the $0$ well.

\section{Coherence Properties}
\label{sec:coherence}

We follow the standard formalism \cite{Shnirman2002, Martinis2003, Schoelkopf2003, Ithier2005, Koch2007, Groszkowski2018} to compute coherence properties of the current-mirror qubit with respect to relevant noise channels. In the usual context of Bloch-Redfield theory \cite{Wangsness1953,Geva1995}, decoherence rates for a given noise channel $\lambda$ are quantified in terms of the pure-dephasing time $T_{\phi}^{\lambda}$, the depolarization time $T_{1}^{\lambda}$, and the dephasing time $T_{2}^{\lambda}$ given by $1/T_{2}^{\lambda}=1/2T_{1}^{\lambda}+1/T_{\phi}^{\lambda}$. The analysis proceeds by a perturbative treatment of noise terms $\delta \lambda(t)$ in the Hamiltonian, leading to the Taylor-series expansion of the effective Hamiltonian,
\begin{equation}
\begin{split}
\label{noiseexpand}
H_{\text {eff}}[\lambda(t)] = H_{\text {eff}}(\lambda_0)+
(\partial_\lambda H_{\text{eff}})
\delta\lambda(t)+ 
\tfrac{1}{2}(\partial_\lambda^{2}H_{\text {eff}})
\delta\lambda^{2}(t),
\end{split}
\end{equation}
where $\lambda(t)=\lambda_{0}+\delta\lambda(t)$, and derivatives are evaluated at $\lambda=\lambda_0$ both here and in the following.

\subsection{Pure Dephasing}
In the context of superconducting qubits, three main noise channels $\lambda$ are likely candidates to dominate pure-dephasing rates: charge noise via offset-charge variations on each node, critical-current noise of each junction, and flux noise in the magnetic flux penetrating the M\"obius ring. Analytical expressions for the pure dephasing times due to these $1/f$ noise channels are given, for example, in Refs.\  \cite{Ithier2005,Groszkowski2018}:
\begin{align}
\label{Tphi}
T_{\phi}^{\lambda}=&\bigg[2A_{\lambda}^2
    \left(D_{0\pi}^{\lambda,(1)}\right)^2
    |\ln{\omega_{ir}t}|\\ \nonumber 
    &+2A_{\lambda}^4
    \left(D_{0\pi}^{\lambda,(2)}\right)^2
    \left(\textstyle \ln^2{\frac{\omega_{uv}}{\omega_{ir}}}+
    2\ln^2{\omega_{ir}}\right)
    \bigg]^{-1/2},
\end{align}
where
\begin{align}
\label{eq:Dmatelem}
D_{0\pi}^{\lambda,(n)}=& \textstyle
    \mel**{0}{\partial^n_\lambda H_{\text {eff}}}{0}-
    \mel**{\pi}{\partial^n_\lambda H_{\text {eff}}}
    {\pi}.
\end{align}
$\omega_{\text{ir}}$ and $\omega_{\text{uv}}$ denote the low- and high-frequency cutoffs of the noise power spectrum, respectively, $t$ represents the timescale of a typical measurement, and $A_{\lambda}$ is the $1/f$ amplitude of the noise power spectrum at positive frequencies \cite{Wellstood1987, Zorin1996, VanHarlingen2004, Yoshihara2006, Pourkabirian2014,Kumar2016, Hutchings2017}
\begin{equation}
\label{S1f}
S_{\lambda}^{1/f}(\omega)=\frac{2 \pi
A_{\lambda}^2}{\omega^{\gamma}},\quad
{\text{ for }} \omega_{ir}<\omega<\omega_{uv},
\end{equation}
where $\gamma\approx1$. We make the common assumption that different noise channels are statistically independent \cite{Ithier2005, Shnirman2002, Yan2016, Martinis2003, Manucharyan2009, Clerk2010}, and use the parameter values $\omega_{\text{uv}}/2\pi= 3.0\, {\text{GHz}}, \omega_{\text{ir}}/2\pi = 1\, {\text{Hz}}$ and $t=10\,\mu$s \cite{Yan2016,Hutchings2017,Quintana2017}. Derivatives of the Hamiltonian appearing in Eq.~\eqref{eq:Dmatelem} are calculated numerically using a five-point stencil, as analytical evaluation of derivatives is prevented by the difficulty of obtaining a closed, analytical form for the degeneracy-breaking amplitude $K$. For charge noise and critical-current noise, the derivatives are performed by varying the noisy quantity at a single site/junction, keeping all other offset charges/junction critical currents constant.
\begin{figure*}
    \centering
    \includegraphics[width=1.0\textwidth]{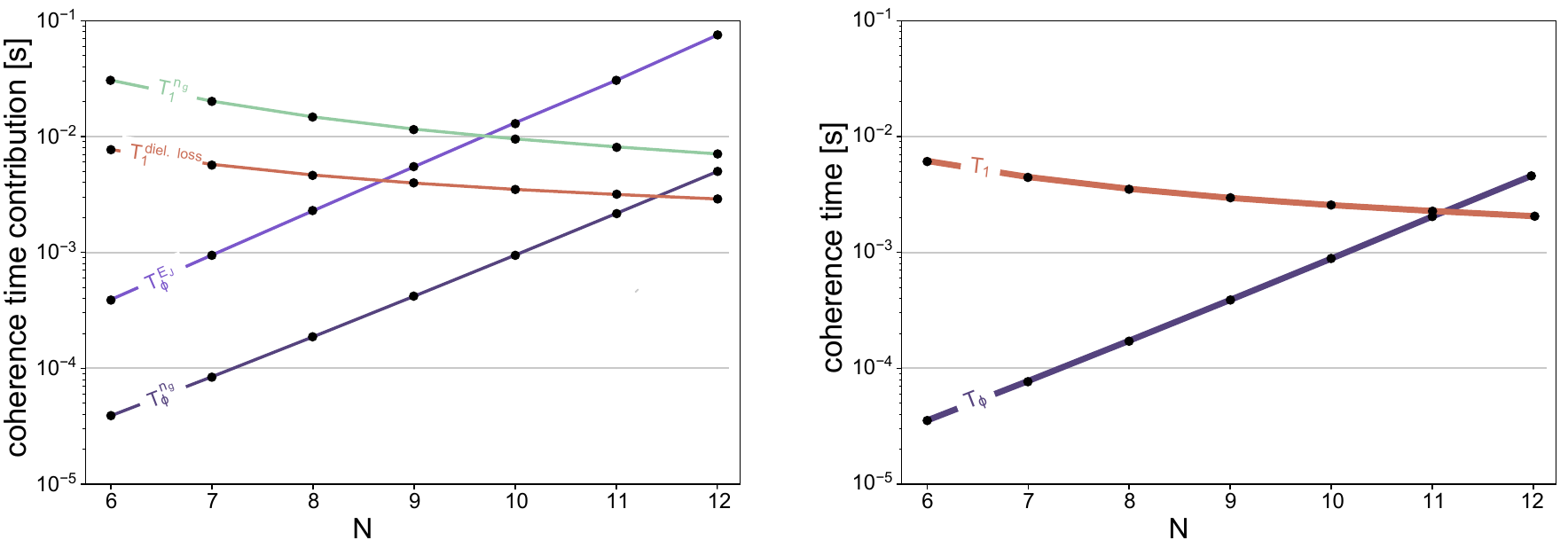}
    \caption{(a) Pure dephasing and depolarization times $T_{\phi}^{\lambda}$ and $T_1^{\lambda}$, as a function of circuit size $N$. Charge noise limits $T_{\phi}$, but dephasing times increase exponentially with $N$. Dielectric loss limits $T_1$, with escape from the qubit subspace dominating over relaxation. Due to the $1/N$ scaling of low-energy excitations thermal excitations become more prominent with increasing circuit size. This explains the observed decrease in $T_{1}^{\lambda}$ as a function of $N$. (b) Total $T_{1}$ and $T_{\phi}$ as a function of circuit size $N$.}
    \label{fig:Tfig}
\end{figure*}
Figure \ref{fig:Tfig} presents our numerical results for the pure-dephasing times of the current-mirror qubit.  We find that charge noise is the limiting factor for pure dephasing, which is not unexpected given that the current mirror operates at $E_J < E_{C_J}$ -- a regime where offset-charge dependence is relevant. Critical-current noise is subdominant, and dephasing due to flux noise is so insignificant that its contributions are outside the range displayed in Fig.\ \ref{fig:Tfig}. Dephasing times overall improve as a function of $N$ because the relevant derivatives $\partial K/\partial n_{gj}$ and $\partial K/\partial E_{J}$ are suppressed as a function of $N$, as a direct consequence of the exponential suppression of $K$ with $N$. (In fact, operating the qubit at $\phi_{\text{ext}}=\pi$ improves offset-charge and critical-current pure-dephasing times by factors of ten or more. However, this operating point is not attractive since protection from relaxation is lost.)

We predict dephasing times on the order of milliseconds for $N$ as small as 10, which will represent a ten-fold increase from the current state-of-the-art \cite{Kjaergaard2019}. We expect the exponential decrease of $K$ to persist for higher values of $N$, indicating that pure dephasing times on the order of tens of milliseconds should be possible for $N=13,14$.

\subsection{Depolarization}
An additional important merit of the current-mirror qubit is its built-in protection from relaxation. Because of the virtually disjoint support of the  $\ket{0}$ and $\ket{\pi}$ wave functions, all matrix elements $\mel{0}{M}{\pi}$ with respect to local operators $M$ are exponentially small. This implies that depolarization in the form of escape upwards to higher-energy eigenstates outside of the qubit subspace is the main contributor to $T_{1}$, as opposed to relaxation within the computational subspace. We will study the same three noise channels considered in the previous subsection, in addition to dielectric loss, which is a known contributor to relaxation \cite{Smith2019,Wang2015}. The total depolarization time $T_{1}^{\lambda}$ due to channel $\lambda$ is obtained by summing over individual depolarization rates, 
\begin{align}
\label{T1}
1/T_{1}^{\lambda}&=1/T_{1}^{\lambda,\pi\rightarrow0}+1/T_{1}^{\lambda,0\rightarrow\pi} \\ \nonumber
&\quad +\sum_{m\ge 2}\left[1/T_{1}^{\lambda,0\rightarrow m} + 1/T_{1}^{\lambda,\pi\rightarrow m}\right],
\end{align}
dominated by transitions to eigenstates outside of the qubit manifold. We find the individual depolarization times via Fermi's golden rule \cite{Ithier2005,Clerk2010,Schoelkopf2003},
\begin{equation}
\label{equation:depolarization}
1/T_1^{\lambda,m\rightarrow n}=\textstyle\frac{1}{\hbar^2}|\langle m |M| n\rangle |^2 S_{\lambda}(\omega_{mn}),
\end{equation}
where $M$ is an operator that induces depolarization and $\omega_{mn}$ is the energy splitting between eigenstates $m$ and $n$ divided by $\hbar$. For noise described via Eq.\ \eqref{noiseexpand}, the operator takes the form $M=\partial_{\lambda}H_{\text {eff}}$. 

Following Refs.~\cite{Smith2019,Wang2015}, we model dielectric loss as dissipation in the dielectric of each capacitor in the current mirror. The operator $M$ involved in depolarization due to dielectric loss is the charge stored on each capacitor \cite{Smith2019}. To find expressions for the charges on the big capacitors and the junction capacitors, we return to the Lagrangian picture. The charge across the $j$-th junction is $\frac{\Phi_{0}}{2\pi}C_{J}(\dot{\phi}_{j+1}-\dot{\phi}_{j})$. Similarly, the charge across the $j$-th big capacitor is given by $\frac{\Phi_{0}}{2\pi}C_{B}(\dot{\phi}_{j+N}-\dot{\phi}_{j})$. In order to evaluate matrix elements, these expressions must be rewritten in terms of operators associated with the effective Hamiltonian. The relation
\begin{equation}
\dot{\phi}_{i}^{-}=\sum_{j=1}^{N}
    8(\mathsf{E_{C}^{-}})_{i,j}n_{j}^{-}/\hbar
    =8E_{C_B}n_{i}^{-}/\hbar +
    \mathcal{O}\left(\textstyle\frac{C_{J}}{C_{B}}\right)
\end{equation}
allows us to recast the capacitor charges in terms of exciton charge operators. The matrix elements in Eq.~\eqref{equation:depolarization} associated with dielectric loss can now be evaluated numerically within the framework of the effective Hamiltonian. The final ingredient for predicting depolarization times due to dielectric loss is the form of the noise power spectrum for a capacitor with capacitance $C$. This is given by \cite{Smith2019,Pop2014}
\begin{equation}
\label{Sdielloss}
S_{\text{diel.}}(\omega,C)+S_{\text{diel.}}
(-\omega,C)=\frac{2\hbar}{C Q(C)}\coth{\frac{\hbar\omega}{2k_{\text{B}}T}},
\end{equation}
where $T$ is the temperature and $Q(C)$ is the quality factor of the dielectric. We use values of $Q(C_{J})=10^6$ and $Q(C_{B})=10^7$ as well as $T=15$ mK \cite{Smith2019,Pop2014}.

While the symmetrized noise power spectrum is useful for discussing relaxation, we are interested mainly in escape from the qubit subspace. These processes only involve the noise power spectrum at negative frequencies, corresponding to absorption of energy from the environment. Assuming the microscopic origin of the noise is a system in thermal equilibrium, the spectrum must obey detailed balance, i.e.,
\begin{align}
\label{Sexpsuppr}
S(-\omega)=S(\omega)e^{-\hbar\omega/(k_{B}T)}.
\end{align}
allowing us to solve for $S_{\text{diel.}}(-\omega,C)$ using Eq.~\eqref{Sdielloss}. Equation \eqref{Sexpsuppr} also implies an exponential suppression of the $1/f$ noise power spectrum [Eq.~\eqref{S1f}] at negative frequencies relative to positive frequencies. Since low-energy excitations in the current mirror scale as $1/N$, we expect  $T_{1}$ to decrease as a function of $N$. Since $T_{\phi}$ was observed to increase with $N$, we expect there to be an optimal $N$ for operating the current-mirror qubit where $T_1$ and $T_\phi$ are of the same order of magnitude. 

We present our results for the depolarization times of the qubit in Fig.~\ref{fig:Tfig}. Dielectric loss is the limiting factor for depolarization at all $N$, which is reasonable given that a circuit of size $N$ by definition has $2N$ junction capacitors and $N$ big capacitors. $1/f$ charge noise is sub-dominant, and contributions from $1/f$ critical-current noise and magnetic-flux noise to depolarization are safely negligible. Our calculations yield depolarization times of multiple milliseconds, representing a ten- or hundred-fold increase over current state-of-the-art transmons, and on-par with current state-of-the-art fluxonium qubits \cite{Kjaergaard2019}. As seen in Fig.~\ref{fig:Tfig}, past $N=11$ the qubit ceases to be $T_{\phi}$ limited and becomes $T_{1}$ limited.

We emphasize that escape from the qubit subspace is the dominant contributor to depolarization, and relaxation within the qubit subspace is vastly suppressed. Such escape processes are only relevant for transitions inside of each well, because of a similar suppression of matrix elements between states in different wells. Interestingly, if the qubit degree of freedom could be made insensitive to harmonic excitations and merely be linked to overall occupation in the $0$ vs.\ the $\pi$ well, then $T_1$ times would be dramatically longer than in our above estimates. 

\section{Conclusion}
\label{concsec}

Kitaev's current-mirror circuit is an attractive qubit concept for intrinsic protection from noise and corresponding long  coherence times, rendering it an interesting design for the next generation of superconducting qubits. Detailed analysis of the current-mirror circuit faces new challenges not common for previously studied superconducting circuits, in particular the significant increase in the number of degrees of freedom. As a consequence, simulating the full Hamiltonian using exact diagonalization has proven difficult, if not intractable for $N>3$ because of memory requirements. To overcome this obstacle, we have presented an effective Hamiltonian describing the low-energy subspace of the current-mirror circuit, halving the number of degrees of freedom and reducing the range of interactions. This effective Hamiltonian was derived by treating the Josephson tunneling perturbatively, resulting in second-order exciton tunneling as well as $N$-th order degeneracy-breaking processes. We have provided a systematic discussion of the degeneracy-breaking terms, crucial to predicting the behavior of the current-mirror qubit. The effective Hamiltonian thus obtained is amenable to DMRG treatment, and has allowed us to simulate circuits with up to $N=12$ big capacitors. For the DMRG numerics, computation time rather than memory poses the bottleneck, and therefore simulation of even larger circuits is possible.

Our numerical DMRG results confirm the development of (approximate) ground-state degeneracy for circuit sizes exceeding $N=6$, tracing the origin of the near-degeneracy to an effective double-well structure with slight asymmetry in the $N$-dimensional configuration space of the effective model. Linearization of the potential around both minima further yields a good approximation for  the circuit's low-energy excitations in terms of harmonic normal modes. An important insight from this analysis is the observed $1/N$ scaling of the energies of low-lying excitations. As a consequence, excessively large circuit sizes $N$ must be avoided as increasing size will eventually lead to depolarization from thermal excitations of the harmonic modes.

Based on the spectral data from DMRG, we have estimated coherence times for the current-mirror qubit for a representative set of parameters. $1/f$ charge noise, critical-current noise, flux noise and dielectric loss were investigated for their contributions to both pure dephasing and depolarization. We find that charge noise is likely to limit $T_{\phi}$, while dielectric loss limits $T_{1}$ in our analysis. $T_{\phi}$ is observed to improve as a function of $N$ because of the decreasing degeneracy-breaking coefficient $K$, while $T_{1}$ worsens as a function of $N$ because of the energy suppression of low-lying eigenstates, and resulting thermal excitations. $T_{\phi}$ and $T_{1}$ nearly intersect at $N=11$, indicating that $N=11$ may be considered an ideal operating point of the qubit for the studied parameter set. Coherence times calculations were performed at the charge sweet spot, leaving open for now the characterization in the presence of offset-charge jumps $>0.1e$ \cite{Christensen2019}. 

Future research will be devoted to modes of control and readout of the current-mirror circuit, as well as the study of quasiparticle tunneling, a mechanism for relaxation and dephasing known to affect superconducting qubits \cite{Lutchyn2005,Serniak2018,Pop2014,Catelani2011,Catelani2011a}. The analysis of the latter is complicated by the fact that  the effective model breaks down at points of maximal charge frustration, and would require a simulation of the full model. However, the presence of long-range interactions mediated by the agiton charging energies challenges DMRG convergence, and we defer a discussion of full-model DMRG results to a future publication.

\section*{Acknowledgements}
The authors thank K.\ R.\ Colladay, A.\ Di Paolo, R.\ J.\ Epstein, Z.\ Huang, W.\ C. Smith, M.\ E.\ Weippert, and X.\ You for valuable discussions. D.K.W.\ was supported in part by an ARO QuaCGR fellowship. This research was supported by the U.S.\ Army Research Office under contract number W911NF-17-C-0024. 

\appendix

\section{Analytical Inverse of the Capacitance Matrix}
\label{appendix:ChargeEn}

To obtain analytical expressions for the inverse capacitance matrix of the current-mirror circuit, we bring $\mathsf{C}$ [Eq.\ \eqref{CapMatrix}] into block-diagonal form, where each of the two blocks is a symmetric, tridiagonal, Toeplitz matrix. Once in this form, we employ results from References \onlinecite{Kershaw1969} and \onlinecite{Meek1980} which also apply to matrices with anomalous corner elements. 

In terms of the $\pm$ coordinates [Eq.\ \eqref{phipm}], the capacitance matrix takes on the form
\begin{equation}
\widetilde{\mathsf{C}}=\frac{1}{2}
 \begin{tikzpicture}[
    baseline={([yshift=-.8ex]current bounding box.center)},
    every left delimiter/.style={xshift=.75em},
    every right delimiter/.style={xshift=-.75em}
  ]
    \matrix[
      matrix of math nodes,
      left delimiter=(,
      right delimiter=),
      nodes in empty cells,
    ] (m) {
	C_{+} & 0 & -C_{J} &   & -C_{J} &  \\
    0 & C_{-} &   &  &  & C_{J}  \\
    -C_{J} &  &   &  &  &        \\
     &  &  &  &  & -C_{J} \\
    -C_{J} &   &    &  & C_{+} & 0 \\
      & C_{J} &  &  -C_{J} & 0 & C_{-}\\
    };
    \draw[loosely dotted, thick] (m-1-2) -- (m-5-6);
    \draw[loosely dotted, thick] (m-1-3) -- (m-4-6);
    \draw[loosely dotted, thick] (m-2-1) -- (m-6-5);
    \draw[loosely dotted, thick] (m-3-1) -- (m-6-4);
    \draw (m-2-2) -- (m-5-5);
  \end{tikzpicture},
\end{equation}
where the ordering of basis vectors is $(1;+),(1;-),\cdots,\,(N;+),(N;-)$. Diagonal dots have the usual meaning, and ``$\diagmatdown$" implies that the diagonally preceding pattern is to be repeated. Further, we have defined $C_{+}=C_{g}+2C_{J}$, $C_{-}=C_{g}+2C_{J}+2C_{B}$. Reordering according to $(1;+),\,(2;+),\cdots,\,(N;+),\,(1;-),\,(2;-),\cdots,\,(N;-)$, one achieves
the block-diagonal form
\begin{equation}
\widetilde{\mathsf{C}}=\left(\begin{matrix}
						\mathsf{C}_{+} & 0 \\
						0 & \mathsf{C}_{-}
					\end{matrix}\right)
\end{equation}
with
\begin{equation}
\mathsf{C}_\pm=\frac{C_{J}}{2}
 \begin{tikzpicture}[
    baseline={([yshift=-.8ex]current bounding box.center)},
    every left delimiter/.style={xshift=.75em},
    every right delimiter/.style={xshift=-.75em},
  ]
    \matrix[
      matrix of math nodes,
      left delimiter=(,
      right delimiter=),
      nodes in empty cells,
    ] (m) {
	x_\pm & -1 &   & \mp 1 \\
 	-1 &  &  &   \\
    &  &  & -1 \\
    \mp 1 &   & -1 & x_\pm\\
    };
    \draw[loosely dotted, thick] (m-1-1) -- (m-4-4);
    \draw[loosely dotted, thick] (m-1-2) -- (m-3-4);
    \draw[loosely dotted, thick] (m-2-1) -- (m-4-3);
  \end{tikzpicture},
\end{equation}
and $x_\pm=C_\pm/C_{J}$. The matrix $\mathsf{C}_{+}$ is cyclic tridiagonal, and is readily inverted \cite{Kershaw1969}, while the anomalous corner elements in $\mathsf{C}_{-}$ require additional work \cite{Meek1980}. One finds
\begin{equation}
(\mathsf{C}_{\pm}^{-1})_{j,k}=\frac{\mp U_{k-j-1}(x_{\pm}/2) -U_{N-k+j-1}(x_{\pm}/2)
}{C_{J}[1-T_{N}(x_{\pm}/2)]},
\end{equation}
where $T_n$, $U_n$ denote the Chebyshev polynomials of the first and second kind.

Based on these results, we determine asymptotic expressions for the charging-energy matrix elements  
\begin{equation}
\label{ecpm}
\left(\mathsf{E_{C}^{\pm}}\right)_{j,k}=\frac{e^2}{2}(\mathsf{C}_{\pm}^{-1})_{j,k},
\end{equation}
associated with the agiton and exciton charges $n_{j}^{\pm}$. In the parameter regime of interest, capacitances follow the hierarchy  $C_{B}\gg C_{J}>C_{g}$. Consequently, matrix elements of $\mathsf{E_{C}^{+}}\sim 1/C_g$ tend to be large compared to relevant elements of  $\mathsf{E_{C}^{-}}\sim 1/C_B$. Agiton charging energies are long-ranged, with maximum entries along the diagonal and monotonically decreasing towards minimum entries along the $N/2$-th off-diagonal. Asymptotic expressions for the diagonal, and the first and $N/2$-th off-diagonal are given by:
\begin{align}\label{ec+1}
E_{C0}^{+}\equiv&\left(\mathsf{E_{C}^{+}}\right)_{j,j} \\ \nonumber =&\frac{e^2}{NC_{g}}\left[1+\frac{C_{g}(N^2-1)}{12C_{J}}+  \mathcal{O}\left(\textstyle\left\{\frac{C_{g}}{C_{J}}\right\}^2\right) \right] , \\\label{ec+2}
E_{C1}^{+}\equiv&\left(\mathsf{E_{C}^{+}}\right)_{j,j\pm1} \\ \nonumber =&\frac{e^2}{NC_{g}}\left[1+\frac{C_{g}(N^2-6N+5)}{12C_{J}} + \mathcal{O}\left(\textstyle\left\{\frac{C_{g}}{C_{J}}\right\}^2\right)\right],
\end{align}
and
\begin{align}\label{ec+3}
E_{C\frac{N}{2}}^{+}&\equiv(\mathsf{E_{C}^{+}})_{j,j\pm\frac{N}{2}} \\ \nonumber
            &=\frac{e^2}{NC_{g}}\left[1-\frac{C_{g}}{24C_{J}}\left(N^2+2 \right) 
             + \mathcal{O}\left(\left\{\textstyle\frac{C_{g}}{C_{J}}\right\}^2\right)   \right].
\end{align}
For excitons, charging energies are much shorter-ranged with off-diagonal elements decreasing rapidly in powers of $C_{J}/C_{B}$. The primarily relevant entries of $\mathsf{E_{C}^{-}}$ are along the diagonal and first off-diagonal,
\begin{align}
\label{ec-1}
E_{C0}^{-}\equiv& \left(\mathsf{E_{C}^{-}}\right)_{j,j} =\frac{e^2}{2C_{B}}+\mathcal{O}\left(\textstyle \frac{C_{J}}{C_{B}}, \frac{C_{g}}{C_{B}}\right), \\
E_{C1}^{-}\equiv& \left(\mathsf{E_{C}^{-}}\right)_{j,j\pm1}
= \frac{e^2C_{J}}{4C_{B}^2} +\mathcal{O}\left(\textstyle \frac{C_{J}}{C_{B}}, \frac{C_{g}}{C_{B}}\right),
\end{align}
and anomalous corner elements $\left(\mathsf{E_{C}^{-}}\right)_{1,N}= -E_{C1}^{-}$.

\section{Leading-Order Processes for Creation and Annihilation of Excitons}
\label{appendix:KNsteps}

Creation and annihilation of an odd number of excitons constitutes the leading-order mechanism for breaking the degeneracy between the two potential minima in the effective exciton model. In this appendix, we establish that leading-order processes occur at $N$-th order in perturbation theory.

\begin{proposition}
Degeneracy between the two potential minima at $\{\phi_{j}^{-}=0\}$ and  $\{\phi_{j}^{-}=\pi\}$ is broken by perturbative processes that create or annihilate an odd number of excitons. Perturbative processes that leave the exciton number invariant or change it by an even number do not lead to degeneracy breaking.
\end{proposition}
\begin{proof}
Consider a perturbative process creating or annihilating $m$ excitons at positions $j_1,\ldots,j_m\in\{1,\ldots,N\}$ with exciton signs specified by $s_1,\ldots,s_m\in \{-1,+1\}$. This process contributes a term to the effective Hamiltonian with operator content
\begin{equation}
A=\prod_{k=1}^m e^{i s_k \phi_{j_{k}}^{-}} + \text{h.c.} = 2\cos\left(  \sum_{k=1}^m s_k \phi_{j_{k}}^{-} \right).
\end{equation}
Addition of $A$ to the Hamiltonian amounts to a modification of the potential energy. For even exciton number $m$, the cosine argument is zero at $\{\phi_{j}^{-}=0\}$ and an even integer multiple of $\pi$ at $\{\phi_{j}^{-}=\pi\}$, thus changing the two potential minima equally and leaving the degeneracy intact. By contrast, for odd exciton number the cosine argument at $\{\phi_{j}^{-}=\pi\}$ is an odd integer multiple of $\pi$, thus leading to an overall sign change between the potential-minima shifts at $\{\phi_{j}^{-}=0\}$ and $\{\phi_{j}^{-}=\pi\}$, effectively breaking the degeneracy.
\end{proof}

We will next prove two central statements. First, the leading order for creation or annihilation of an odd number of excitons is order $N$ (where $N$ is the number of big capacitors in the circuit). Second, every such $N$-th order process resulting in odd-number changes in exciton population leads to charge alternation: a `$+$' exciton is always followed by a `$-$' exciton, so that when circling the edge of the M\"obius circuit, a Cooper-pair charge is always followed by a Cooper-pair hole, and vice versa. An example of this is 
\begin{align*}
    A_1 &= e^{-i \phi_{2}^{-}}e^{+i \phi_{3}^{-}}e^{-i \phi_{N}^{-}} + \text{h.c.}\\
        &= e^{-i \phi_{2}}e^{+i \phi_{3}}e^{-i \phi_{N}}e^{i \phi_{N+2}}e^{-i \phi_{N+3}}e^{i \phi_{2N}} + \text{h.c.},
\end{align*}
as shown in Fig.\ \ref{excitoncreation}. The 3-exciton creation process described by
$B=e^{i\phi_{1}^{-}}e^{i\phi_{3}^{-}}e^{i\phi_{7}^{-}} + \text{h.c.}$, on the other hand, does not obey charge alternation and is of order higher than $N$.

We first prove that odd-number exciton creation with charge alternation requires an $N$-th order process. To assess the minimal order of the perturbative term for exciton creation, we note that the process separates positive and negative charges and moves them in such a fashion to ultimately recover an exciton configuration. Each step of moving a charge along the circuit circumference is achieved by an operator from the perturbing Hamiltonian $H_J$, such as $e^{i\phi_j}e^{-i\phi_{j+1}}$, and increases the order of the perturbative process by one.

Our proof showing that order $N$ is the minimum required order relies on mapping our problem to a special instance of the so-called assignment problem known from combinatorial optimization \cite{Kuhn1955,Karp1975,Werman1986}, formulated as follows. Consider two ordered sets $M=\{m_{1},m_{2}, \cdots,m_{n}\}$ and $P=\{p_{1},p_{2},\cdots,p_{n}\}$ which here denote the $n$ positions of minus and plus charges on the circuit. Each minus charge is generated by charge separation and increasing the relative difference to some plus charge. The perturbative order of the creation process is thus ascertained by assigning each minus charge to a plus charge and adding up their ``spatial'' separations. The order of a process is equal to the cost $C$ of a particular assignment, given by
\begin{equation}
C = \sum_{i,j}C({m_i,p_j})X_{ij}.
\end{equation}
Here, assignments are recorded by the $n{\times}n$ permutation matrix $X$ with  $X_{ij}=1$ if $m_{i}$
is assigned to $p_{j}$; otherwise, $X_{ij}=0$. Distance between positions on the circle is measured by
\begin{equation}
C({m_i,p_j})=\min\Big[(m_i-p_j) \text{ mod } 2N, \,(p_j- m_i) \text{ mod } 2N\Big].
\end{equation}
Determining the minimum perturbative order required to achieve the desired exciton creation thus corresponds to finding the optimal assignment $X$ which minimizes the cost $C$. We first show that nearest-neighbor assignment on the circle for creation of an odd number of excitons obeying charge alternation leads to a cost of $N$, and subsequently prove that this assignment is optimal. (Hence, order $N$ is the leading order for odd-number exciton creation.)

\begin{proposition}
\label{nearestneighbor}
For creation of an odd number of excitons obeying charge alternation, nearest-neighbor assignment has cost $N$.
\end{proposition}
\begin{proof}
The proof proceeds by induction over the (odd) number $m$ of excitons. For the base case $m=1$ (a single exciton trivially obeys charge alternation), there is only one Cooper-pair creation operator and one annihilation operator. The two generated charges are nearest neighbors with distance $N$, so the cost of the only possible assignment is $N$. 

Next, assume that nearest-neighbor assignment indeed has a cost of $N$ for $m=n$ excitons alternating in sign, and show that the same is true for $m=n+2$ excitons. To do so, decompose the creation operator $A_{n+2}$ for $n+2$ excitons into creation of $n$ excitons,
\begin{align}
A_n&=
e^{-i\phi_{i_{1}}^{-}}e^{i\phi_{i_{2}}^{-}}\cdots
e^{-i\phi_{i_{n}}^{-}} \\\nonumber
&= e^{-i\phi_{i_{1}}}e^{i\phi_{i_{2}}}\cdots e^{-i\phi_{i_{n}}}
	e^{i\phi_{i_{1}+N}}e^{-i\phi_{i_{2}+N}}\cdots e^{i\phi_{i_{n}+N}},
\end{align}
and creation of two additional excitons,  see Fig.\ \ref{fig:add2excitons}. (``$\text{h.c.}$" contributions are omitted from expressions to simplify notation.) For $A_n$, there are two different nearest-neighbor assignments with equal cost $N$: either pairing up $e^{-i\phi_{i_{1}}}$ and $e^{i\phi_{i_{2}}}$, 
or $e^{i\phi_{i_{n}+N}}$ and $e^{-i\phi_{i_{1}}}$. 
\begin{figure}
\includegraphics[width=1.0\columnwidth]{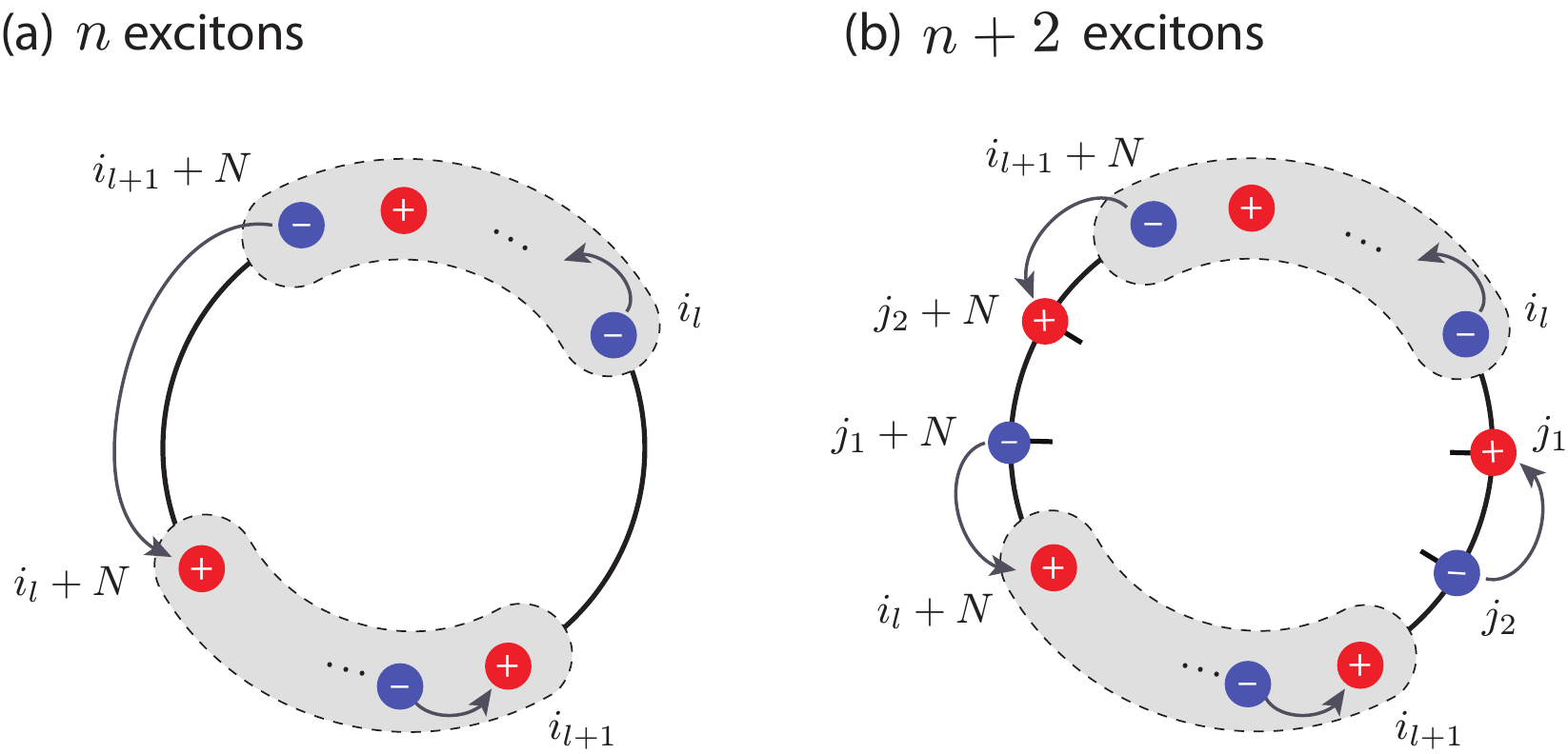}
\caption{(a) Nearest-neighbor assignment for an odd number $n$ of excitons with alternating signs. For odd $n$, the assignment of $i_{l+1}+N$ to $i_l+N$ implies that the partner charges $i_{l+1}$ and $i_l$ are not assigned to each other (or vice versa). (b) Creation of $n+2$ excitons with alternating signs is obtained from creation of $n$ and inserting two nearest neighbor excitons with the appropriate signs. The cost for nearest-neighbor assignment remains $N$.
\label{fig:add2excitons}}
\end{figure}
To maintain sign alternation, the additional two excitons on rungs $j_{1}$ and 
$j_{2}$ must be nearest neighbors, and $A_{n+2}$ has the form
\begin{equation*}
A_{n+2}=
e^{-i\phi_{i_{1}}^{-}}\cdots 
e^{-i\phi_{i_{l}}^{-}}e^{+i\phi_{j_{1}}^{-}}e^{-i\phi_{j_{2}}^{-}}
	e^{+i\phi_{i_{l+1}}^{-}}\cdots e^{-i\phi_{i_{n}}^{-}},
\end{equation*}
where the insertion point is between $i_l$ and $i_{l+1}$.

Without loss of generality, let us assume that for $A_n$, $i_l+N$ and $i_{l+1}+N$ are paired. Since $n$ is odd, this implies that $i_{l}$ and $i_{l+1}$ are not paired. Insertion of the two additional excitons then leads to the nearest-neighbor assignment shown in Fig.\ \ref{fig:add2excitons}(b). The new cost of this assignment can be read off from the figure and is given by
\begin{align*}
    N' &= N - (i_{l+1}+N - i_l -N) + (i_{l+1} + N - j_2 -N) \\
    &\quad +(j_1+N - i_l - N) + (j_2-j_1) = N,
\end{align*}
thus confirming that creation of $n+2$ excitons with alternating signs also carries cost $N$.
\end{proof}

While nearest-neighbor assignment thus leads to a cost of $N$, it remains to be shown that $N$ is the minimum possible cost. The situation is straightforward for the assignment problem on the line, which is known to be solved by a greedy assignment, $m_{i} \rightarrow p_{i}$ \cite{Karp1975}. The assignment problem on the circle, which we face here, requires more thought, see Ref.~\cite{Werman1986}. Therein, Werman et al.\ show that the circular assignment problem can be reduced to the linear one by identifying an appropriate cutting point. Once this point is used for cutting the circle, linear greedy assignment minimizes the cost. As a corollary to this general result, we can therefore state for our case:
\begin{corollary}
The optimal assignment for creation of an odd number
of excitons with alternating signs consists of nearest-neighbor assignment.
\end{corollary}
\begin{proof}
This follows from the work by Werman et al., and from the fact that greedy assignment for charges with alternating signs results in nearest-neighbor assignment.
\end{proof}

Note that the nearest-neighbor assignment for an odd number of excitons obeying charge alternation leads to either a clockwise or a counter-clockwise assignment, i.e., all assignment arrows pointing from minus charges to plus charges are oriented clockwise or oriented counter-clockwise. A counter-clockwise assignment is shown in Fig.~\ref{fig:add2excitons}(b). Since the cost of both assignments is $N$, and we understand cost as order of perturbation theory, both contribute at $N$-th order.

Finally, we show that creation of an odd number of excitons not obeying sign alternation has an optimal cost strictly larger than $N$. To facilitate the proof, we require some additional notation borrowed from Ref.~\onlinecite{Werman1986}. Given the ordered sets $M = \{m_{i}\}$ and $P = \{p_{i}\}$ for locations of minus and plus charges, we define
\begin{equation}
\begin{split}
F_{m}(x)=&\left| \{i:\,m_{i}<x\} \right|,\\
F_{p}(x)=&\left| \{i:\,p_{i}<x\} \right|.
\end{split}
\end{equation}
Here, $F_m(x)$ counts the number of minus charges between the origin and position $x$ on the circle; likewise $F_p(x)$ does so for plus charges. The difference 
\begin{equation}
F(x) = F_{p}(x)-F_{m}(x),
\end{equation}
quantifies the net positive charge between the origin and location $x$. $F$ is a piecewise constant function with discontinuities at charge locations. For an alternating charge configuration, $F$ alternates between either $0$ and $+1$, or $0$ and $-1$, such that $f=\max_x F(x)-\min_{x'} F(x')=1$. For non-alternating charge configurations, $f$ exceeds 1. 

\begin{figure}
\includegraphics[width=0.75\columnwidth]{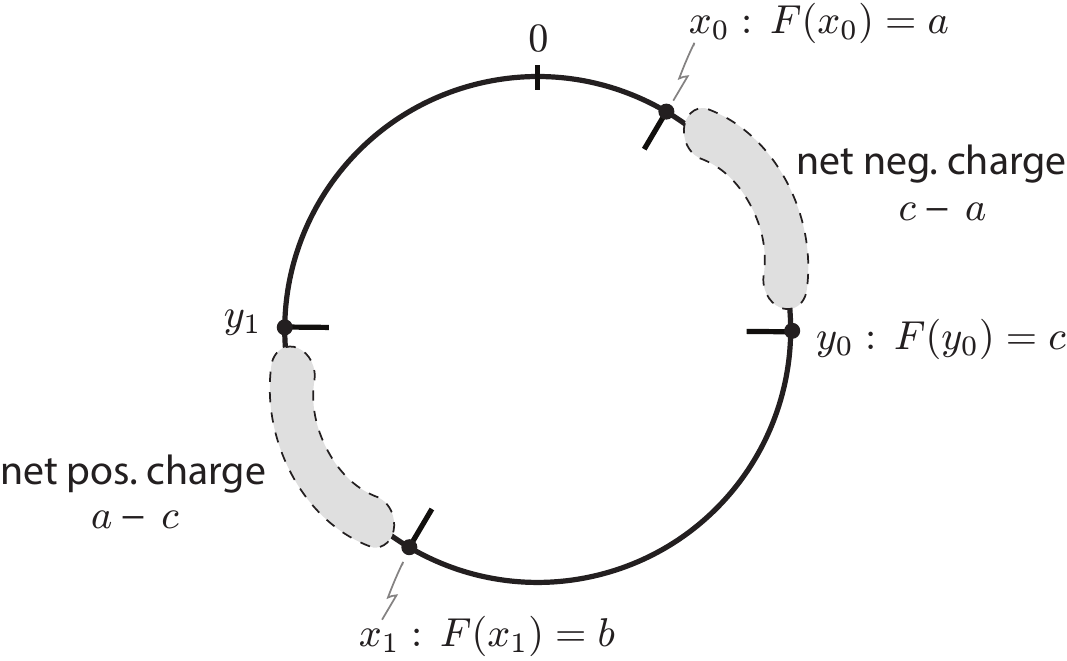}
\caption{Visual aid for the proof of Lemma 1, showing that the point antipodal to $F(x_0)=\text{max}$ must be a point where $F=\text{min}$.}
\label{minmax}
\end{figure}

The following three lemmas prove instrumental in the proof that excitons with non-alternating signs require a higher cost.
\begin{lemma}
Consider a configuration of an odd number of excitons on the circle and let $x_{0}$ be a position where $F$ is maximal. Then the antipodal point $x_{1}=x_{0}+N$ marks a position where $F$ is minimal.
\end{lemma}
\begin{proof}
The proof is by contradiction and aided visually by Fig.\ \ref{minmax}. Let $F(x_0)=a=\text{max}$, and let $F(x_{1})=b$ at the antipodal point $x_{1}=x_{0}+N$. Now assume that there exists some other point $y_0$ with $x_0 < y_0 < x_1$  that yields an $F$ even smaller: $F(y)=c<b$. (An analogous argument holds for $ y_0 > x_1$.) Then, the interval $[x_0,y_0]$ contains a net negative charge $c-a<0$. Let $y_1=y_0+N$ be the point antipodal to $y_0$. Due to the exciton configuration of charges, the interval $[x_1, y_1]$ must contain the net positive charge $a-c>0$. As a result, we have 
\[
F(y_1) = F(x_1) + a - c = a + (b-c) > a = F(x_0),
\]
in contradiction to the maximality of $F(x_0)$.
\end{proof}

The following lemma states an important property of the points where $F$ is maximal or minimal. This property regards the assignment arrows (also called ``arcs") above these points:
\begin{lemma}
Consider an odd number of excitons not obeying sign alternation. In the optimal assignment obtained from the algorithm by Werman et al., points of maximal and minimal $F$ must have arcs above them.
\end{lemma}
For proof of this lemma the reader is referred to Ref.~\onlinecite{Werman1986}.

\begin{figure}
\includegraphics[width=0.9\columnwidth]{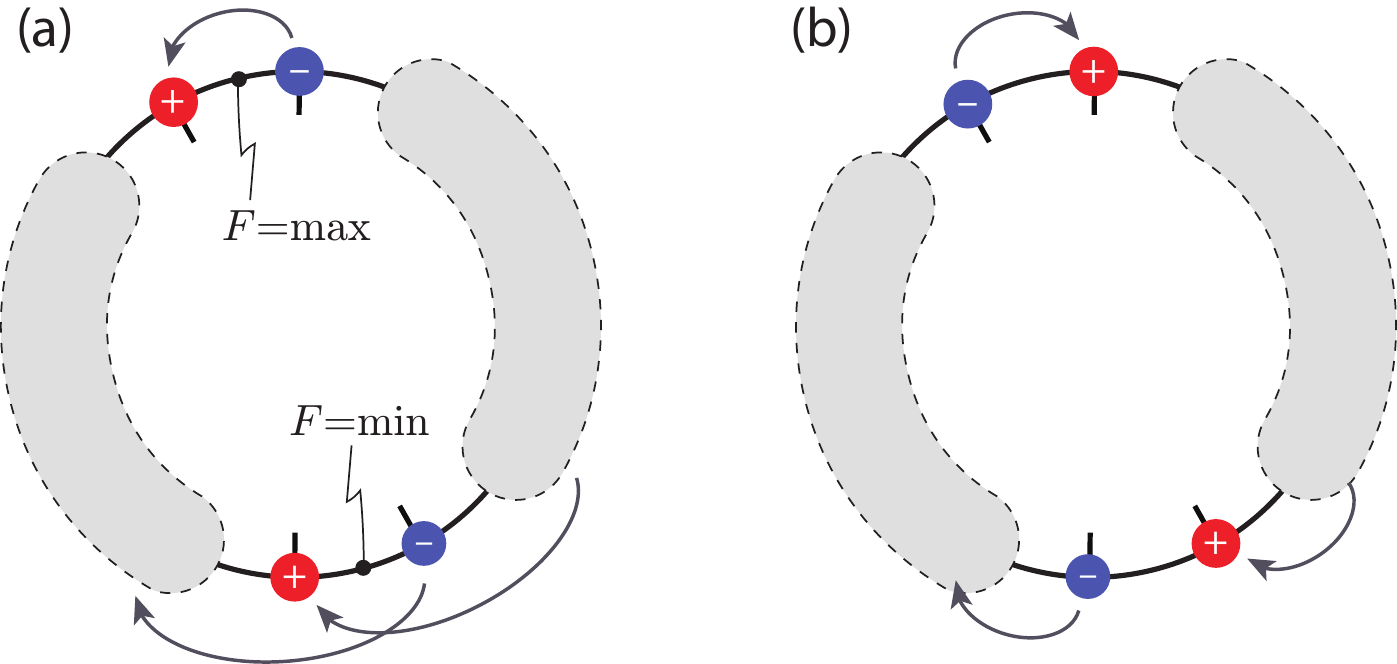}
\caption{(a) Configuration of charges adjacent to points of maximal and minimal $F$. Arrows show one of the possible scenarios of optimal assignments. (b) New charge configuration and optimal assignment obtained after swapping the locations of the excitons from (a). The resulting assignment has a  cost strictly less
than that in (a).}
\label{gen_case}
\end{figure}

\begin{lemma}
Consider once more an odd number of excitons not obeying sign alternation and antipodal points where $F$ is maximal and minimal, see Fig.\ \ref{gen_case}(a). Then the shaded region reached from $F=\text{max}$ traveling clockwise must contain net negative charge.
\end{lemma}
\begin{proof}
First, note that the charges surrounding the points of maximum and minimum $F$ must have signs as indicated in Fig.\ \ref{gen_case}(a). Since the overall exciton number is odd, the shaded region of interest must contain an odd number of charges, so the net charge $a$ in that region is odd. Further, we have 
\[ F_\text{min}= F_\text{max} + a - 2 < F_\text{max}. \]
From this, we conclude $a<2$. We rule out $a=1$, as this would imply $f=F_\text{max}-F_\text{min}=1$ and hence charge alternation. $a=0$ is not possible since $a$ is odd. Hence, we have $a<0$, so the net charge in the region of interest is negative.
\end{proof}

These lemmas are now utilized in the proof of the following statement about non-alternating exciton configurations.
\begin{proposition}
For configurations of an odd number of excitons not obeying sign alternation, the cost of the optimal assignment is strictly larger than $N$.
\end{proposition}
\begin{proof}
The proof is by induction on $f=\max_x F(x)-\min_{x'} F(x')$, starting with the base case $f=3$. According to Lemma 3, the points of minimal and maximal $F$ must have arcs overhead in the optimal assignment. Since this assignment is obtained by greedy pairing, nested arcs cannot occur. This leaves only two possibilities of optimal assignments for the charges adjacent to the points with extremal $F$. The first is shown in Fig.~\ref{gen_case}(a), where the top two charges are assigned to each other and the bottom ones have crossing arcs. The second possibility places crossing arcs both on the top and bottom pair of charges. By contrast, direct pairing of both the top two charges and the bottom two charges does not yield an optimal assignment as can be seen as follows. Suppose the top two charges are paired. Lemma 3 asserts that the shaded region on the right of Fig.\ \ref{gen_case}(a) contains net negative charge. Since minimal cost is achieved by greedy assignment in clockwise fashion, the plus charge on the bottom must be assigned to a negative charge in this shaded region. The negative charge on the bottom must thus be assigned to a positive charge in the shaded region on the left, leading to crossing arcs as shown in the figure.

We assume an optimal assignment of the type depicted in Fig.~\ref{gen_case}(a);  the following arguments also carry over to the case of crossing arcs for both segments. Locate all positions where $F=F_\text{max}$ and swap the positions of the two adjacent excitons. This changes the maximum and minimum values of $F$ to $F_\text{max}'=F_\text{max}-1$ and $F_\text{min}'=F_\text{min}+1$, and hence yields $f'=f-2$. For $f=3$, the charge swaps thus produce an exciton configuration with sign alternation. The resulting assignment, shown in Fig.\ \ref{gen_case}(b) has a cost strictly lower than the assignment in the non-alternating case of Fig.\ \ref{gen_case}(a), completing the proof for the base case of $f=3$.

For the induction step, assume that any odd-number exciton configuration with $f\leq f_{0}$ has an optimal-assignment cost strictly greater than $N$, and show that this is true as well for $f=f_{0}+2$. (Note that $f$ can only take on odd-integer values for an odd number of excitons.) The argument is analogous to that employed for the base case $f=3$. Identifying locations of maximal $F$ and swapping the adjacent excitons, one finds a new configuration with a strictly lower cost and $f'=f-2=f_{0}>1$. By the inductive hypothesis, this has an optimal-assignment cost strictly larger than $N$, and so the assertion is proven.
\end{proof}

\section{Derivation of the Effective Hamiltonian}
\label{appendix:Nthorder}
In this appendix, we sketch the derivation of terms relevant for the effective exciton Hamiltonian. The procedure is based on a Schrieffer-Wolff transformation, closely following  Ref.~\onlinecite{Shavitt1980}. The starting point is the full circuit Hamiltonian $H= H_{C}+H_{J}$ with charging terms
\begin{align}
H_C = \sum_{i,j=1}^{2N}4(\mathsf{E_{C}})_{ij}
    (n_{i}-n_{gi})(n_{j}-n_{gj})
\end{align}
and Josephson tunneling terms
\begin{align}
H_J =-\frac{E_{J}}{2} \sum_{j=1}^{2N} e^{i(\phi_{j+1}- \phi_j)}e^{-i\phi_\text{ext}/2N} + \text{h.c.}
\end{align}

 The unperturbed Hamiltonian $H_C$ divides the Hilbert space into a low-energy subspace $\alpha$ and a high-energy subspace $\gamma$ -- the former spanned by low-lying exciton charge states (eigenstates of $n_j^-$), the latter by agiton charge states (eigenstates of $n_j^+$). The corresponding unperturbed eigenenergies and eigenstates are denoted $E_{i,\alpha},\,E_{j,\gamma}$ and $|i,\alpha\rangle,\,|j,\gamma\rangle$. Cooper-pair tunneling  acts as a perturbation $H_J=V_X+V_D$, coupling the two energy manifolds via its block-off-diagonal component 
\[
V_{X}=\sum_{i,j}|i,\alpha\rangle \langle i,\alpha |H_J| j,\gamma\rangle \langle j,\gamma| + \text{h.c.},
\]
and  individual states within the high-energy manifold via its block-diagonal part 
\[
V_{D}=\sum_{j,j'}|j,\gamma\rangle \langle j,\gamma |H_J| j',\gamma\rangle \langle j',\gamma|.
\]
(Matrix elements between states in the low-energy subspace vanish.) 
 
Using a Schrieffer-Wolff transformation, we construct an effective Hamiltonian $H'=e^{-G}He^{G}$ which incorporates second-order exciton-hopping terms, and $N$-th order degeneracy-breaking terms. (All other terms of order smaller than $N$ cannot break degeneracy and are omitted.) The generator $G$ of the unitary transformation is anti-Hermitian and purely block off-diagonal. Systematic construction of $G$ and of the resulting low-energy Hamiltonian \[ P_\alpha H' P_\alpha = P_\alpha H_0 P_\alpha + W^{(2)} + W^{(N)} \] proceeds via expansion in $H_J$ and iterative employment of the Baker-Campbell-Hausdorff relation, see Ref.\ \onlinecite{Shavitt1980}. To second order, this yields the exciton-hopping terms
\begin{align}
\label{eq:W2}
W^{(2)}&=\frac{1}{2}\ket{i,\alpha}\bra{i,\alpha}
        H_{J}\ket{k,\gamma}\bra{k,\gamma}H_{J}
        \ket{j,\alpha}\bra{j,\alpha}\\\nonumber
        &\quad\times\left(\frac{1}{E_{i,\alpha}-E_{k,\gamma}}+
        \frac{1}{E_{j,\alpha}-E_{k,\gamma}}\right),
\end{align}
where here and in the following, summation over repeated Latin indices is implied. In this perturbative path for exciton hopping, a single high-energy virtual state $\ket{k,\gamma}$ is accessed. Calculation of the energy denominators in Eq.~\eqref{eq:W2} in principle depends on the states $\ket{i,\alpha}$ and $\ket{j,\alpha}$ and their energies $E_{i,\alpha}$ and $E_{j,\alpha}$, respectively. However, for $N\lesssim20$ the smallest charging-energy matrix element of the agiton coordinates, $E_{C\frac{N}{2}}^{+}$, is much larger than the largest charging-energy matrix element of the exciton coordinates, $E_{C0}^{-}$. Therefore, we may neglect exciton charging energies in the calculation of Eq.~\eqref{eq:W2}, leading to the simplified expression
\begin{align*}
W^{(2)}=\ket{i,\alpha}\bra{i,\alpha}
        H_{J}\ket{k,\gamma}\frac{1}{-E_{k,\gamma}}\bra{k,\gamma}H_{J}
        \ket{j,\alpha}\bra{j,\alpha}
\end{align*}
The virtual state $\ket{k,\gamma}$ accessed depends on the initial state $\ket{j,\alpha}$. For instance, consider an exciton tunneling from rung $\ell$ to rung $\ell+1$. The initial state is given by $\ket{\{n_{i}^{+}=0\},\{n_{i}^{-}\}}$. There are two possible virtual agiton states, namely
\[
\ket{a\pm}=\ket{
\begin{array}{cc}	
n_{\ell}^{+}=\mp\frac{1}{2}, & n_{\ell+1}^{+}{=}\pm\frac{1}{2}\\
n_{\ell}^{-}{=}n_{\ell}-\frac{1}{2}, & n_{\ell+1}^{-}{=}n_{\ell+1}+\frac{1}{2}
\end{array}
}
\]
with energies $E_{a}^{\pm}$ that can be accessed. We thus find 
\begin{widetext}
\begin{align}
W^{(2)}=\sideset{}{'}\sum_{\ell=1}^{2N}\bigg( &\ket{i,\alpha}\bra{i,\alpha}
        e^{i\phi_{\ell+1}}e^{-i\phi_{\ell}}e^{-i\phi_{\text{ext}}/2N} \ket{a-}\frac{1}{-E_{a}^{-}}\bra{a-}e^{-i\phi_{\ell+N+1}}e^{i\phi_{\ell+N}}e^{i\phi_{\text{ext}}/2N}\ket{j,\alpha}\bra{j,\alpha}  \\ \nonumber
        + &\ket{i,\alpha}\bra{i,\alpha} e^{-i\phi_{\ell+1}}e^{i\phi_{\ell}}e^{i\phi_{\text{ext}}/2N}\ket{a+} \frac{1}{-E_{a}^{+}}\bra{a+} e^{i\phi_{\ell+N+1}}e^{-i\phi_{\ell+N}}e^{-i\phi_{\text{ext}}/2N}
        \ket{j,\alpha}\bra{j,\alpha}\bigg),
\end{align}
\end{widetext}
where it is clear now that the external flux drops out exactly. When neglecting exciton energies as compared to agiton coordinates, the intermediate-state energies become independent of the exciton quantum numbers: 
\begin{align*}
E_{a}^{\pm}&=2\Delta E_{j}^{\pm} 
    =2(E_{{\text C}0}^{+}-E_{{\text C}1}^{+}) \\ \nonumber
    &\pm 4\left[(\mathsf{E_{C}}^{+})_{m,j}-(\mathsf{E_{C}}^{+})_{m,j+1}\right]
    n_{gm}^{+}+\mathcal{O}(E_{C0}^{-}),
\end{align*}
see Eqs.\ \eqref{ecpm}--\eqref{ec-1} for definitions of charging energies involved. This simplification allows for the sum over initial states indexed by $j$ to be performed, yielding
\begin{align}\nonumber
W^{(2)}=&-\sum_{j=1}^{N-1}\frac{E_{J}^2}{4}
	\left(\frac{1}{\Delta E_{j}^{+}}
	    +\frac{1}{\Delta E_{j}^{-}}\right) \cos(\phi_{j+1}^{-}-
	\phi_{j}^{-})\\ 
	&-\frac{E_{J}^2}{4}
	\left(\frac{1}{\Delta E_{N}^{+}}
	    +\frac{1}{\Delta E_{N}^{-}}\right)\cos(\phi_{1}^{-}+\phi_{N}^{-}),
\end{align}
where the identification $e^{i\phi_{j}^{-}}=\sum_{n}\ket{n_{j}^{-}=n+1}\bra{n_{j}^{-}=n}$ has been used and terms merely introducing energy renormalization of charging energies are omitted. The resulting exciton-hopping strengths are given by
\begin{equation}
\label{eq:JDef}
J_{j}=\frac{E_{J}^2}{4}
	\left(\frac{1}{\Delta E_{j}^{+}}
	    +\frac{1}{\Delta E_{j}^{-}}\right).
\end{equation}
For vanishing agiton offset charges, $n_{gj}^{+}=0$, the exciton-hopping strengths simplify to the uniform expression
\begin{equation}
J=\frac{E_{J}^2}{2E_{{\text C}0}^{+}-2E_{{\text C}1}^{+}}
       =\frac{E_{J}^2}{2E_{C_{J}}}+\mathcal{O}\left(\frac{C_{g}}{C_{J}},\frac{1}{N}\right).
\end{equation}

Next, we consider the leading-order process for degeneracy-breaking of the potential minima in the effective model: the creation or annihilation of an odd number of excitons with sign alternation. Including only these degeneracy-breaking terms, we find the relevant $N$-th order contribution 
\begin{align}
\nonumber
W^{(N)}&=\frac{1}{2}|i,\alpha\rangle\langle i,\alpha|\prod_{\ell=1}^{N-1}(H_{J}|k_\ell,\gamma\rangle
    \langle k_\ell,\gamma |) H_{J} |j,\alpha\rangle \langle j,\alpha | \\  &\times\Bigg[
		\frac{1}{\prod_{\ell=1}^{N-1}(E_{i,\alpha}-E_{k_\ell,\gamma})}
	+\frac{1}{\prod_{\ell=1}^{N-1}(E_{j,\alpha}-E_{k_\ell,\gamma})} \Bigg]
	\label{NPT}
\end{align}
where $H_{J}$ appears $N$ times and there are $N-1$ energy denominators representing the cost of accessing virtual states from the high-energy manifold. Based on this expression and the results from Appendix~\ref{appendix:KNsteps}, we find that $W^{(N)}$ reduces to the Hamiltonian
\begin{equation}
\label{eq:KOpp}
H_{\text{K}} = -K\cos \frac{\phi_{\text{ext}}}{2} \sum_{m\leq N}^{\text{odd}}
\sum_{i_{1}<i_{2}<\cdots<i_{m}}
\cos \bigg[ \sum_{j=1}^{m}(-1)^{j}\phi_{i_{j}}^{-} \bigg].
\end{equation}
Here, $K$ is the amplitude for creation and annihilation of an odd number $m$ of excitons with sign alternation. The inner sum runs over all ordered sequences of an odd number of rung indices $i_{n}$, indicating the positions of excitons. Simple combinatorics reveals the number of terms in the inner sum of Eq.~\eqref{eq:KOpp} as follows. Given a number of exciton creation/annihilation terms $m$, there are $N-m$ choices of where to place the remaining empty rungs. Therefore there are $\binom{N}{m}$ terms in the inner sum of Eq.~\eqref{eq:KOpp} for each $m$. Summing over all of these contributions yields $\sum_{m\le N}^\text{odd} \binom{N}{m}=2^{N-1}$  \cite{gradshteyn2007}. 

Given the various perturbative paths and corresponding energy denominators associated with Eq.\ \eqref{NPT} and contributing to Eq.\ \eqref{eq:KOpp}, it is not immediately clear that all amplitudes can be approximated by the same constant $K$. While it is clear that  $K\sim (E_{J}/2)^{N}$, the computation and approximation of energy denominators is more involved. To do so, we must track the high-energy virtual states accessed in the perturbative paths.

We illustrate the procedure for the perturbative paths contributing to the creation of a single exciton on rung 1, associated with the operator $\exp(i\phi_{1}^{-})$. According to Appendix \ref{appendix:KNsteps}, the relevant $N$-th order perturbative paths involve either exclusively clockwise transfer of Cooper pairs, or exclusively counter-clockwise transfer. The counter-clockwise contributions are summarized by
\begin{widetext}
\begin{align}
\label{eq:A}
A=
&\frac{1}{2}\left(\frac{E_J}{2}\right)^N e^{i\phi_\text{ext}/2} \sum_{p\in S_N} 
|i,\alpha\rangle \langle i,\alpha |
\left(
\prod_{\ell=1}^{N-1} e^{i\phi_{p(\ell)}-i\phi_{p(\ell)+1}} |k_\ell,\gamma\rangle \langle k_\ell,\gamma|\right)
e^{i\phi_{p(N-1)}-i\phi_{p(N-1)+1}}|j,\alpha\rangle \langle j,\alpha|
\\ \nonumber  &\times\Bigg[
		\frac{1}{\prod_{\ell=1}^{N-1}(E_{i,\alpha}-E_{k_\ell,\gamma})}
	+\frac{1}{\prod_{\ell=1}^{N-1}(E_{j,\alpha}-E_{k_\ell,\gamma})} \Bigg],
\end{align}
\end{widetext}
where the summation is over all permutations $p(n)$ of the numbers $1\le n \le N$. An analogous expression is obtained for clockwise perturbative paths.

The energy denominators in Eq.~\eqref{eq:A} are obtained by tracking the high-energy virtual states accessed. The states involved in one virtual process generally differ from those in another, depending on the permutation $p(n)$. Because there are $N!$ such permutations, the difficulty of carrying out this sum rapidly increases with $N$. As in the calculation of Eq.~\eqref{eq:W2}, we neglect exciton-charging energies. This approximation leads to a critical simplification of the energy denominators in Eq.~\eqref{eq:A}; it allows us to sum over each permutation $p(n)$ where the initial state is a circuit devoid of charge, and to ignore exciton charging energies in the intermediate states as compared to agiton energies. The sum in Eq.~\eqref{eq:A} and its clockwise counterpart can be carried out numerically, yielding the expression $\frac{K}{2}\cos(\phi_{\text{ext}}/2)\exp(i\phi_{1}^{-})$.

To confirm that $K$ is the same for all terms in Eq.~\eqref{eq:KOpp}, consider the following. As depicted in Fig.~\ref{fig:add2excitons}, to obtain terms with operator content 
$e^{i(\phi_1^- -\phi_q^- + \phi_r^-)}$ ($r>q$), we needed to sum over terms from Eq.~\eqref{NPT} of exactly the same form as Eq.~\eqref{eq:A}. Here, however, permutations refer to the $N$ integers  $[1\,..\,q-1]\cup[q+N\,..\,r+N-1]\cup[r\,..\,N]$. Forming the operator $e^{i(\phi_1^- -\phi_q^- + \phi_r^-)}$ can be related to the formation of $e^{i\phi_1^-}$ by the following substitutions: $e^{i(\phi_{q} - \phi_{q+1})}\rightarrow e^{i(\phi_{q+N} -\phi_{q+N+1})},\cdots, e^{i(\phi_{r-1} -\phi_{r})} \rightarrow e^{i(\phi_{r+N-1} -\phi_{r+N})}$. To understand the implications of this substitution for the energy denominators, we examine the action of the involved operators on states in the high-energy subspace. Observe that
\begin{align}
&e^{i(\phi_{q} -\phi_{q+1})} \ket{
\begin{array}{cc}
n_{q}^{-}=n_{1}, & n_{q+1}^{-}=n_{2},\\
n_{q}^{+}=m_{1}, & n_{q+1}^{+}=m_{2}
\end{array}
} \\ \nonumber
&\qquad= \ket{
\begin{array}{cc}
n_{q}^{-}=n_{1}+\tfrac{1}{2}, & n_{q+1}^{-}=n_{2}-\tfrac{1}{2},\\
n_{q}^{+}=m_{1}+\tfrac{1}{2}, &n_{q+1}^{+}=m_{2}-\tfrac{1}{2}
\end{array}
},
\intertext{while}
& e^{i(\phi_{q+N} -\phi_{q+N+1})}\ket{
\begin{array}{cc}
n_{q}^{-}=n_{1}, &n_{q+1}^{-}=n_{2},\\
n_{q}^{+}=m_{1}, &n_{q+1}^{+}=m_{2}
\end{array}
}
\\ \nonumber
	&\qquad=\ket{
\begin{array}{cc}	
n_{q}^{-}=n_{1}-\tfrac{1}{2}, &n_{q+1}^{-}=n_{2}+\tfrac{1}{2},\\
n_{q}^{+}=m_{1}+\tfrac{1}{2}, &n_{q+1}^{+}=m_{2}-\tfrac{1}{2}
\end{array}
},
\end{align}
where $n_{i},m_{i}\in \mathbb{Z}$. The crucial insight from this is that the action of the substituted sister operators yields states that have the same agiton charge numbers, and differ only in the exciton charge numbers. Because the term in Eq.~\eqref{eq:KOpp} involving $\exp(i[\phi_{1}^{-}-\phi_{q}^{-}+\phi_{r}^{-}])$ is obtained via a substitution of these sister operators in Eq.~\eqref{eq:A} relative to the term involving $\exp(i\phi_{1}^{-})$, their energy denominators will be identical in the approximation that the exciton charging energies are neglected. Similarly, the operators representing the creation of 5, 7,$\ldots$ excitons that include the creation of an exciton on the first rung must all have the same coefficient as $\exp(i\phi_{1}^{-})$ by the same reasoning. For identical junction and ground capacitances, the rotational symmetry of the circuit is intact. Therefore, the coefficient of $\exp(i\phi_{j}^{-})$ for $j\neq 1$ must be the same as the coefficient of $\exp(i\phi_{1}^{-})$. Therefore all terms in Eq.~\eqref{eq:KOpp} must have the same coefficient.

From numerics for the selected parameter set, we find that $K$ has the functional form $K(N)=175$ GHz $\times$ $\exp(-1.59 N)$, quantifying the exponential suppression at large $N$.
To establish an analytical bound on $K$ we make one further approximation. Ignoring offset charge, we approximate the charging energy of two unpaired charges anywhere on the circuit by $2\Delta E=2E_{{\text C}0}^{+}-2E_{{\text C}1}^{+}$ (and for four unpaired charges by $4\Delta E = 4E_{{\text C}0}^{+}-4E_{{\text C}1}^{+}$, etc.), as opposed to the exact expression $2\Delta E_{j,k} = 2E_{{\text C}0}^{+} - 2(\mathsf{E_{C}^{+}})_{j,k}$ for a Cooper-pair hole on site $j$ and Cooper pair on site $k$. This yields an upper bound for $K$ because $(\mathsf{E_C^{+}})_{j,j\pm n}<E_{C1}^{+}<E_{C0}^{+}$ for $n\ge 2$. By counting all paths and extracting the approximate energy denominators, we find
\begin{equation}
\label{Kbound}
K(N) \leq 4\left(\frac{E_{J}}{2}\right)^N
        \frac{A_{N}}{(\Delta E)^{N-1}},
\end{equation}
where $A_N$ is observed to obey $A_{2}=1$ and $A_{N+1}=A_{N}(2N-1)/N$. This recursion relation can be solved using Pochhammer symbols, yielding
\begin{equation}
A_{N}=\frac{(1)_{2N-3}}{(2N-4)!!(2N-4)!}.
\end{equation}
As discussed in the main text, there are $2^{N-1}$ degeneracy breaking terms in Eq.~\eqref{KNcirc}. Therefore we are interested not just in a bound for $K(N)$, but for $2^{N-1}K(N)$; if $K(N)$ decreases slower than $1/2^{N-1}$, then a ground-state degeneracy does not develop at large $N$. Using the large $N$ expression $A_{N}\sim 2^{N-2}$, we find
\begin{equation}
\label{Kupperbound}
2^{N-1}K(N) < E_{J}\left(\frac{2E_{J}}{\Delta E}\right)^{N-1}.
\end{equation}
Therefore as long as $2E_{J}<\Delta E$, the degeneracy breaking terms disappear in the large-$N$ limit. The parameters used in this paper yield the energy scales $2E_{J}=38$ GHz and $\Delta E > 60$ GHz for $N\geq 4$, indicating that a ground state degeneracy should indeed develop. 

\bibliography{CM_PAPER_BIB}
\end{document}